\documentclass[%
 reprint,
 amsmath,amssymb,
 aps,
]{revtex4-2}
\usepackage[utf8]{inputenc}

\usepackage[caption=false]{subfig}
\usepackage{graphicx}
\usepackage{here}
\usepackage{epstopdf}
\usepackage{array}
\usepackage{physics}
\usepackage{verbatim}
\usepackage{amsmath,amsfonts,amssymb,amscd,mathtools}
\usepackage{amsthm}
\usepackage{tabularx}
\usepackage{stmaryrd}
\usepackage{enumerate}
\usepackage{wasysym}
\usepackage{braket}
\usepackage{mathrsfs}
\usepackage{microtype}
\usepackage{hyperref}
\usepackage[dvipsnames]{xcolor}
\hypersetup{
    bookmarksnumbered=true, 
    unicode=false, 
    pdfstartview={FitH}, 
    pdftitle={}, 
    pdfauthor={}, 
    pdfsubject={}, 
    pdfcreator={}, 
    pdfproducer={}, 
    pdfkeywords={}, 
    pdfnewwindow=true, 
    colorlinks=true, 
    linkcolor=NavyBlue, 
    citecolor=NavyBlue, 
    filecolor=NavyBlue, 
    urlcolor=NavyBlue 
}
\usepackage{tikz}
\usepackage[lmargin=.7in,rmargin=.7in,tmargin=.7in,bmargin=1in]{geometry}
\usepackage{newtxtext,newtxmath}

\usepackage{cleveref}

\usepackage{algorithm}
\usepackage{algorithmicx}
\usepackage{algpseudocode}



\usepackage{booktabs}

\theoremstyle{plain}
\newtheorem{thm}{Theorem}

\newtheorem{lem}[thm]{Lemma}
\newtheorem{pro}[thm]{Proposition}

\theoremstyle{definition}
\newtheorem{defn}[thm]{Definition}

\newcommand{\eq}[1]{(\hyperref[eq:#1]{\ref*{eq:#1}})}

\renewcommand{\sec}[1]{\hyperref[sec:#1]{Section~\ref*{sec:#1}}}
\newcommand{\thrm}[1]{\hyperref[thrm:#1]{Theorem~\ref*{thrm:#1}}}
\newcommand{\lemm}[1]{\hyperref[lemm:#1]{Lemma~\ref*{lemm:#1}}}
\newcommand{\prop}[1]{\hyperref[prop:#1]{Proposition~\ref*{prop:#1}}}
\newcommand{\corr}[1]{\hyperref[corr:#1]{Corollary~\ref*{corr:#1}}}
\newcommand{\fig}[1]{\hyperref[fig:#1]{~\ref*{fig:#1}}}
\newcommand{\deff}[1]{\hyperref[deff:#1]{~\ref*{deff:#1}}}

\newcommand{\mC}{\mathcal{C}}
\newcommand{\mE}{\mathcal{E}}
\newcommand{\mN}{\mathcal{N}}

\newcommand{\mT}{\mathcal{T}}
\newcommand{\mD}{\mathcal{D}}

\newcommand{\mF}{\mathcal{F}}
\newcommand{\mL}{\mathcal{L}}

\newcommand{\mH}{\mathcal{H}}
\newcommand{\mM}{\mathcal{M}}
\newcommand{\mO}{\mathcal{O}}

\newcommand{\mP}{\mathcal{P}}
\newcommand{\mQ}{\mathcal{Q}}

\newcommand{\mS}{\mathcal{S}}

\newcommand{\mbB}{\mathbb{B}}
\newcommand{\mbC}{\mathbb{C}}

\newcommand{\mbF}{\mathbb{F}}
\newcommand{\mbI}{\mathbb{I}}

\newcommand{\mbM}{\mathbb{M}}
\newcommand{\mbN}{\mathbb{N}}
\newcommand{\mbO}{\mathbb{O}}
\newcommand{\mbR}{\mathbb{R}}
\newcommand{\mbS}{\mathbb{S}}
\newcommand{\mbT}{\mathbb{T}}

\newcommand{\mbZ}{\mathbb{Z}}

\newcommand{\ve}{\varepsilon}

\DeclareMathOperator{\id}{id}

\DeclareMathOperator{\aff}{aff}

\DeclareMathOperator{\Span}{Span}
\DeclareMathOperator{\supp}{supp}
\DeclareMathOperator{\diag}{diag}
\DeclareMathOperator{\spec}{spec}
\DeclareMathOperator{\argmin}{argmin}

\newcommand{\ba}{\begin{eqnarray}}
\newcommand{\ea}{\end{eqnarray}}
\newcommand{\bann}{\begin{eqnarray*}}
\newcommand{\eann}{\end{eqnarray*}}
\newcommand{\bal}{\begin{equation}\begin{aligned}}
\newcommand{\eal}{\end{aligned}\end{equation}}
\newcommand{\dm}[1]{\ketbra{#1}{#1}}
\newcommand{\bs}[1]{\boldsymbol{#1}}
\newcommand{\RomanNumeralCaps}[1]{\MakeUppercase{\romannumeral #1}}

\newcolumntype{L}[1]{>{\raggedright}p{#1}}
\newcolumntype{C}[1]{>{\centering}p{#1}}
\newcolumntype{R}[1]{>{\raggedleft}p{#1}}
\newcolumntype{D}{>{\centering\arraybackslash}X}

\newcommand{\sbar}{\;\rule{0pt}{9.5pt}\right|\;}

\newcommand{\lset}{\left\{\left.}
\newcommand{\rset}{\right\}}

\newcommand{\rt}[1]{{\color{red} #1}}
\newcommand{\bt}[1]{{\color{blue} #1}}

\newcommand{\redunderline}[1]{\textcolor{red}{\underline{\textcolor{black}{#1}}}} 
\newcommand{\blueunderline}[1]{\textcolor{blue}{\underline{\textcolor{black}{#1}}}}

\usepackage{dcolumn}
\usepackage{bm}

\begin{document}

\title{Black box work extraction and composite hypothesis testing}

\author{Kaito Watanabe}
\email{watanabe715@g.ecc.u-tokyo.ac.jp}
\affiliation{Department of Basic Science, The University of Tokyo, 3-8-1 Komaba, Meguro-ku, Tokyo 153-8902, Japan}

\author{Ryuji Takagi}
\email{ryujitakagi.pat@gmail.com}
\affiliation{Department of Basic Science, The University of Tokyo, 3-8-1 Komaba, Meguro-ku, Tokyo 153-8902, Japan}

\begin{abstract}
Work extraction is one of the most central processes in quantum thermodynamics. However, the prior analysis of optimal extractable work has been restricted to a limited operational scenario where complete information about the initial state is given. Here, we introduce a general framework of black box work extraction, which addresses the inaccessibility of information on the initial state. We show that the optimal extractable work in the black box setting is completely characterized by the performance of a composite hypothesis testing task, a fundamental problem in information theory. We employ this general relation to reduce the asymptotic black box work extraction to the quantum Stein's lemma in composite hypothesis testing, allowing us to provide their exact characterization in terms of the Helmholtz free energy. We also show a new quantum Stein's lemma motivated in this physical setting, where a composite hypothesis contains a certain correlation. Our work exhibits the importance of information about the initial state and gives a new interpretation of the quantities in the composite quantum hypothesis testing, encouraging the interplay between the physical settings and the information theory.

\end{abstract}
\maketitle
\paragraph{Introduction.---}

One of the major goals in thermodynamics is to characterize the ultimate efficiency of work extraction. In particular, provided the recent technological developments in accurately controlling nanoscale systems, it is of fundamental importance to obtain a precise understanding of the extractable work from small systems where quantum properties are not negligible. 
Recently, there has been much progress in characterizing extractable work in quantum systems employing quantum information-theoretic approaches~\cite{horodecki2013fundamental, Faist_2018_Fundamental_work_cost, Brandao_2015_second_law,Brandao_2013_RT_of_quantum_states_out_of,Lostaglio2019introductory,Gour_2022_Role_of_quantum_coherence}. 
These results not only provide an explicit form of the optimal single-copy (one-shot) extractable work, but also offer a smooth connection to the many-copy (asymptotic, thermodynamic) limit, where the Helmholtz free energy arises as an emergent quantity~\cite{Brandao_2013_RT_of_quantum_states_out_of, Gour_2022_Role_of_quantum_coherence}.  

Although these results entail fundamental insights into the problem of work extraction, they do not represent natural operational settings. 
Crucially, the optimal work characterized so far assumes that the description of the initial state is provided, allowing the experimenters to tailor the work extraction protocol depending on the given state.  
However, in many settings---such as the scenarios where the state is obtained by a complicated quantum process that cannot be efficiently simulated classically, or the state experiences unknown noise---we are not in possession of the complete information about the initial state. To run the ``state-aware'' protocol in these settings, one would first attempt to learn the state description by quantum state tomography~\cite{Vogel_1989_Determination,Banaszek_2013_focusing}. 
At this point, the characterization of the prior results becomes unclear because (1) state tomography requires multiple (indeed, many) copies of the initial state and thus could significantly change the effective performance of work extraction and (2) the full state tomography may not be possible due to the physical limitation inherent in thermodynamic processes. 
To encompass this large class of ``state-agnostic'' scenarios, new techniques are required.

A major observation from a series of works is that work extraction is closely related to the standard information-theoretic task known as \emph{hypothesis testing}, where one aims to distinguish two quantum states. 
These entirely different-looking operational tasks turn out to be quantitatively connected via their performances.
The maximum amount of work extractable from a single copy of the given known state is precisely characterized by \emph{hypothesis-testing divergence}~\cite{hiai_1991_proper,Wang_One-Shot_Classical-Quantum_Capacity}---the standard quantifier for the asymmetric state discrimination---between the initial (known) state and the thermal Gibbs state~\cite{horodecki2013fundamental, Buscemi_2019_An_information-theoretic,Wang_RT_of_assymmetric_distinguishability, Gour_2022_Role_of_quantum_coherence}. 

Interestingly, hypothesis testing has been extended to a more general setting---instead of distinguishing two states, one aims to distinguish two \emph{sets} of states by measuring a state picked from either of the sets. 
This task is known as \emph{composite hypothesis testing} and has been an active investigation in classical~\cite{Brandao_Adversarial_Hypothesis,Levitan_2002_composite_Neyman_Pearson, Polyanskiy_2013_Saddle_point,fangwei_1996_hypothesis_testing_with_AVS,modak_2023_hypothesis_testing_for_adversarial_channels, Tomamichel_2018_Operational_interpretation} and quantum~\cite{Bjelakovic2005quantum,Brandao2010_A_generalization_of,Notzel2014hypothesis,Berta_2021_On_Composite_HT,Bergh2023composite, hayashi_2016_correlation, gao_2024_generalized} information theory.
In particular, there has been a rising interest in quantum Stein's lemma~\cite{hiai_1991_proper, Ogawa_2000_strong}, which connects composite hypothesis testing divergence to the optimized relative entropy in the asymptotic composite hypothesis testing setting. 
Nevertheless, unlike the case of standard hypothesis testing, the operational significance of composite hypothesis testing in the context of quantum thermodynamics has been unclear.

Here, we establish the fundamental relation between these two---state-agnostic work extraction and composite hypothesis testing. 
We introduce a general framework for state-agnostic work extraction by considering a ``black box'', from which a state is picked and an experimenter---who knows what states are contained in the box but does not know which state was actually picked---applies a work extraction protocol. 
We show that the optimal guaranteed extractable work from a black box can be exactly characterized by the performance of composite hypothesis testing between the black box and thermal Gibbs state. 
This not only extends the result of state-aware work extraction to much more general and operational settings, but also provides the first operational interpretation of composite hypothesis testing in terms of quantum thermodynamics. 

We further employ this relation to obtain the asymptotic work extraction rate in the black box setting. 
Notably, we prove a new kind of Stein's lemma for composite hypothesis testing, where state copies from the composite hypothesis have a correlation generated by a pinching channel~\cite{hiai_1991_proper, Masahito_Hayashi_2002_optimal_sequence, Tomamichel_2016_QI_pricessing_with_finite_resources}. 
This---together with the general connection between black box work extraction and composite hypothesis testing---shows that the asymptotic work extraction rate from a black box with several standard classes of thermodynamic processes~\cite{Faist_2018_Fundamental_work_cost} can be smaller than the minimum Helmholtz free energy of the state in the black box. This implies the fundamental difficulty in the state-agnostic setting compared to the standard setting.
Additionally, we show that a similar characterization can be extended to a class of thermodynamic operations amenable to easier physical implementation~\cite{horodecki2013fundamental}. 

Work extraction protocol in quantum thermodynamics is an example of \emph{quantum resource distillation}.
We extend the notion of the black box resource distillation to general quantum resource theories~\cite{Chitamber_2019_QRT,gour2024resources_of_quantum_world}---examples of which include quantum entanglement~\cite{Horodecki_2009_quantum_entanglement} and magic states~\cite{Veitch_2014_RT_of_stabilizer,Howard_2017_Application}---and show that the optimal performance of the distillation task is universally characterized by the performance of the composite hypothesis testing divergence.
Potential limitations of resource distillation with unknown input states were discussed for some specific cases of entanglement, magic states and work extraction in the context of ergotoropy by different approaches~\cite{Matsumoto_universal_distortion_free_entanglement,aaronson_2023_quantum_pseudoentanglement,gu_2023_little_magic, Safranek_work_extraction}.
Our results complement these findings, offering a platform that allows mutual developments in state-agnostic resource distillation and composite hypothesis testing and boosting the interplay between physically motivated tasks and information-theoretic problems.


\paragraph{Thermodynamic operations.---}
We consider a system associated with a finite-dimensional Hilbert space $\mH$ with some Hamiltonian $H$ in contact with the heat bath of inverse temperature $\beta$.
Here, we employ a resource-theoretic approach to quantum thermodynamics to formalize the set of thermodynamic operations available for work extraction.
The idea of quantum resource theory is to consider the set of states easily prepared in the given physical setting (often called \emph{free states}) and the operations preserving the set of free states as accessible (often called \emph{free operations}). 

In quantum thermodynamics, it is standard to consider the thermal Gibbs state $\tau:=\exp(-\beta H)/\Tr[\exp(-\beta H)]$ as the only free state. Therefore, thermodynamically ``free'' operations must preserve the Gibbs state~\cite{Lostaglio2019introductory}, and several classes of such operations have been investigated depending on the goal of the study. 
The largest class that satisfies the minimum requirement is the \emph{Gibbs-preserving operations}~\cite{Janzing_2000_thermodynamic,Faist_2018_Fundamental_work_cost}, which include all operations that map the thermal state of the input system to that of the output system. 
This class is mainly studied due to its simple mathematical structure, which led to a number of recent key progress in quantum thermodynamics~\cite{Faist_2018_Fundamental_work_cost,Faist_2019_thermodynamic,shiraishi_quantum_thermo,Buscemi_2019_An_information-theoretic, Wang_RT_of_assymmetric_distinguishability, Sagawa2021asymptotic, Liu_2019_one-shot,Regula_2020_Benchmarking}. 
On the other hand, the class that respects the physical implementability is known as \emph{thermal operations}. A completely positive trace-preserving (CPTP) map $\mE$ from systems $A$ to $B$ is called a thermal operation if $\mE$ can be written as 
$\mE(\cdot)=\Tr_{(A+E)\backslash B}\qty[U(\cdot\otimes \tau_E)U^\dagger]$,
where $\tau_E$ is a thermal state of the ancillary system, and $U$ is an energy-conserving unitary satisfying $[U, H_A\otimes I_E+I_A\otimes H_E]=0$. 
Note that the thermal operation is Gibbs-preserving. 

Another important property of the thermal operation is the time translation covariance, i.e., any thermal operation $\mE$ from systems $A$ to $B$ satisfies 
\bal\label{Eq: definition of time translation covariance}
\mE\qty(e^{-iH_At}\rho_A e^{iH_At})=e^{-iH_Bt}\mE(\rho_A)e^{iH_Bt},~~\forall t\in\mbR,
\eal
which prohibits the operation from creating energetic coherence---which serves as another important thermodynamic resource~\cite{Takagi_2022_correlation, Shiraishi_2024_Arbitrary, tajima_2022_universal, tajima_2024_gibbspreserving, Gour_2022_Role_of_quantum_coherence}---from scratch. 

The class of operations that is mathematically easy to handle and closer to thermal operation is called \emph{Gibbs-preserving covariant operations}~\cite{Cwiklinski_limitations_on,Gour2018quantum,Shiraishi2024quantumthermodynamicscoherencecovariant}, which are Gibbs-preserving and time-translation covariant.


\paragraph{Composite hypothesis testing.---}
The composite quantum hypothesis testing aims to distinguish two different sets $\mS$, $\mT$ of states---called a null and an alternative hypothesis ---with a binary measurement with the positive operator-valued measure (POVM) elements $\qty{M, I-M}$ in which the outcome corresponding to $M$ and $I-M$ means that one guesses the given state is an element of $\mS$ and $\mT$ respectively.
If the hypotheses contain a single state, this setting is reduced to ordinary hypothesis testing.

The performance of this task is rephrased as how much one can minimize the probability of making mistakes in the guess. 
 In this work, we mainly focus on type \RomanNumeralCaps{2} error $\sup_{\sigma\in\mT}\Tr[\sigma M]$---the minimum worst-case probability of guessing the state in $\mT$ as that in $\mS$---under the constraint in which type \RomanNumeralCaps{1} error $\sup_{\rho\in\mS}\Tr[\rho(I-M)]$---the worst-case probability of guessing the state in $\mS$ as that in $\mT$---is at most $\ve$. The performance of this task is represented as the quantity called the hypothesis testing divergence defined as~\cite{Brandao2010_A_generalization_of, Bergh2023composite, Berta_2021_On_Composite_HT}
\bal
D_H^\ve(\mS||\mT)=-\log \inf_{\substack{0\leq M\leq I\\\sup_{\rho\in \mS}\Tr[\rho(I-M)]\leq \ve}}\sup_{\sigma\in\mT}\Tr[\sigma M].
\eal
If $\mT$ is a singleton $\{\tau\}$, we simply write $D_H^\ve(\mS\|\tau)$ without the set notation. 
Note that this quantity does not change even if we take the convex hull of either of the composite hypothesis.


\paragraph{Black box work extraction.---}
We now introduce the framework of black box work extraction.
In addition to the main system, we consider another system called a ``battery'' associated with a 2-dimensional Hilbert space $\mH_X=\Span\qty{\ket{0},\ket{1}}$. 
Following Ref.~\cite{Gour_2022_Role_of_quantum_coherence}, we take the Hamiltonian for the battery system as $H_{X}=E_{X,0}\dm{0}+E_{X,1}\dm{1}$ with $E_{X,1}-E_{X,0}=\beta^{-1}\log(m-1)$ so that the thermal state of the battery system $\mu_m$ is $\mu_m=\frac{m-1}{m}\ketbra{0}{0}+\frac{1}{m}\ketbra{1}{1}$ for $m\geq 1$.
If an allowed operation can transform an initial state and the equilibrium state $\mu_m$ in the battery system to the state $\dm{1}_X$ of the battery, we say that we can ``charge'' the battery.

We represent the inaccessibility to the information of the given state as a black box, a subset $\mS\subset\mD(\mH)$ of the set $\mD(\mH)$ of all quantum states.
The experimenters are informed about the description of $\mS$ and that the initial state is picked from $\mS$ but are not told which state is  given, preventing them from tailoring work extraction protocols depending on the state.  
Note that the state is not picked according to some known probability distribution. If so, the situation is reduced to the state-aware scenario where one is given the averaged state taken over the black box.

The problem is to find the maximum $m$ such that the battery with the thermal state $\mu_m$ can be charged with the unknown initial state picked up from the black box and allowed operations $\mbO$, i.e., to find the largest $m$ such that $\rho\otimes \mu_m\to\ketbra{1}{1}_X$ is possible for every choice of the state from the black box and the allowed operation which is independent of the initial state $\rho\in \mS$. 
Here, we formulate the optimal performance of the black box work extraction
\begin{defn}
    The one-shot extractable work of the black box $\mS\subset\mD(\mH)$ with error $\ve$ is defined  as $\beta W_{\mbO}^\ve(\mS)=\log m^*$, where $m^*$ is the largest $m\in\mbR$ which satisfies 
    \bal
    \max_{\mE\in\mbO} \min_{\rho\in\mS}F(\mE(\rho),\ketbra{1}{1}_X)\geq 1-\ve.
    \eal
    Here, $\mu_m=\frac{m-1}{m}\ketbra{0}{0}+\frac{1}{m}\ketbra{1}{1}$ is the Gibbs state of the battery system $X$, and $F(\rho,\sigma)=\| \sqrt{\rho}\sqrt{\sigma} \|_1^2$ is the square fidelity. 

\end{defn}
For the justification of this definition, see Appendix A~\cite{note_for_reference}.
\nocite{sion_1958_general_minimax_theorems,eisert_2003_remarks,Rubboli_private_communication,Rubboli_2024_new,Rubboli_2023_mixed_state_additivity,noauthor_probabilistic_nodate,Piani_2009_relative_entropy_of, Cover_Thomas, Wang_2019_resource, Lostagilo_Quantum_coherence_time_translation_symmetry, O'Donnell2016efficient, Takagi_2019_general_resource, Gour_2019_how_to, Liu_2019_resource_theories, Liu_2020_operational_resource, Regula_2021_fundamental_limitations, Fang_2022_no-go_theorems, Kuroiwa_2024_robustness, Bennett_teleporting_unknown, Bennett_communication_via, Bennett_1996_concentrating, Bravyi_2005_universal, Winter_2016_operational, Marvian_2020_coherence, Wang_2019_quantifying, Gour_2020_dynamical_entanglement, Takagi_2020_application, Kim_2021_One-shot_manipulation, Takagi_2022_One-shot_yield-cost_relation, Vidal_robustness_of_entanglement, Datta_2009_min_and_max_relative}

One can also consider the asymptotic limit of this by considering a sequence $\qty{\mS_n}_{n=1}^\infty$ of the black boxes with $\mS_n\subset \mD(\mH^{\otimes n})$. Suppose there are $n$ systems with the Hamiltonian $H^{\times n}:=\sum_{j=1}^nI^{\otimes j-1}\otimes H\otimes I^{\otimes (n-j)} $. 
To take the limit $n\to \infty$, we consider a family $\qty{\mS_n}_{n=1}^\infty$ of black boxes with $\mS_n\subset \mD(\mH^{\otimes n})$.
We define the asymptotic black box extractable work as the work drawn from the whole system per the number of subsystems, namely,
    $
    \beta W_\mbO(\qty{\mS_n}_{n=1}^\infty):=\lim_{\ve\to +0}\limsup_{n\to\infty}\beta W^\ve_\mbO(\mS_n)/n.
    $

In the main text, we mainly focus on the sequence of the black boxes with a tensor-product structure $\mS^{\rm TP}_n(S):=\qty{\bigotimes_{i=1}^n\rho_i~|~\rho_i\in S, \forall i}$ generated by an arbitrary set $S\subset\mD(\mH)$. In Appendix C, we show that similar results can be obtained for a slightly more general sequence of the black boxes.

\paragraph{Black box work extraction with Gibbs-preserving operations.---}
We are in the position to characterize the performance of black box work extraction. 
We first consider Gibbs-preserving operations as available thermodynamic processes.
The following result provides the general characterization of one-shot extractable work in terms of composite hypothesis testing divergence (Proof in Appendix B.1). 
\begin{thm}\label{theorem: one-shot GPO}
Let $H$ and $\tau$ be the Hamiltonian and the thermal state of the system. One-shot extractable work from an arbitrary black box $\mS$ under Gibbs-preserving operations (GPO) satisfies
    \bal\label{Eq: p-one shot extractable work under GPO}
    \beta W^\ve_{\rm{GPO}}(\mS)&=D^\ve_H(\mS||\tau).
    \eal
\end{thm}
Theorem~\ref{theorem: one-shot GPO} establishes a tight connection between the composite hypothesis testing and the work extraction task and provides a physical meaning of the composite hypothesis testing divergence in the context of thermodynamics. 
In the case of a singleton set $\mS=\{\rho\}$, our result recovers the known result for state-aware work extraction~\cite{Wang_RT_of_assymmetric_distinguishability,Gour_2022_Role_of_quantum_coherence}.

Let us now extend this to asymptotic work extraction. 
Theorem~\ref{theorem: one-shot GPO} allows us to focus on whether the composite hypothesis testing divergence connects to the standard relative entropy under the asymptotic limit, a central question in information theory known as Stein’s lemma.
This comes with a further physical significance in the context of quantum thermodynamics because relative entropy precisely corresponds to the free energy, which plays a central role in thermodynamics. 

Despite the difficulty, previous works found that there are several settings in which composite Stein's lemma can be established~\cite{Bjelakovic2005quantum,Brandao_Adversarial_Hypothesis,Notzel2014hypothesis,Berta_2021_On_Composite_HT,Mosonyi2022on,Berta2022_on_a_gap,Bergh2023composite,gao_2024_generalized}.
In particular, an extension of quantum Sanov's theorem~\cite{Notzel2014hypothesis}, together with our general characterization in Theorem~\ref{theorem: one-shot GPO}, implies the following simple expression for the asymptotic work extraction. (Proof in Appendix C.1)

\begin{thm}\label{theorem: asymptotic GPO}
Let $H$ and $\tau$ be the Hamiltonian and the thermal state of the system. The asymptotic black box extractable work of the sequence $\qty{\mS^{\rm TP}_n(S)}_{n=1}^\infty$ of the black boxes under Gibbs-preserving operations (GPO) is given by
\bal\label{Eq: p-asymptotic extractable work under GPO}
\beta W_{\rm{GPO}}(\qty{\mS^{\rm TP}_n(S)}_{n=1}^\infty)&=\min_{\rho\in\mC(S)}D(\rho||\tau),
\eal
where $D(\rho\|\tau)\coloneqq \Tr(\rho\log\rho)-\Tr(\rho\log\tau)$ is the Umegaki relative entropy, and $\mC(S)$ is the convex hull of $S$. 
\end{thm}

We remark that if the sequence of the black boxes is composed of i.i.d. states, i.e., $\mS^{\rm{i.i.d.}}_n(S):=\qty{\rho^{\otimes n}~|~\rho\in S}$, the right-hand side of Eq.~\eqref{Eq: p-asymptotic extractable work under GPO} is reduced to $\min_{\rho\in S}D(\rho||\tau)$, which can also be seen as a consequence of the quantum Sanov's theorem~\cite{Bjelakovic2005quantum}.

Theorem~\ref{theorem: asymptotic GPO} clarifies the fundamental restriction imposed by not knowing the input state. 
In Appendix C.3, we exhibit examples of the sequence of black boxes, which reveals the underlying difference between the standard state-aware work extraction task and the black box work extraction task.

\paragraph{Black box work extraction under Gibbs-preserving covariant operations.---}
Although Gibbs-preserving operations admit relatively simple mathematical analysis, there is also doubt in its operational justification. 
Notably, they can create quantum coherence from scratch~\cite{Faist2015Gibbs-preserving}, and some Gibbs-preserving operations require even unbounded quantum coherence to implement~\cite{tajima_2024_gibbspreserving}.
This motivates us to impose additional constraints described in \eqref{Eq: definition of time translation covariance} that operations should be time-translation covariant---which prohibits the creation and detection of quantum coherence---and this is precisely the class of Gibbs-preserving covariant operations.

We consider  the pinching channel, an important analytical tool in information theory~\cite{Tomamichel_2016_QI_pricessing_with_finite_resources,hiai_1991_proper,Masahito_Hayashi_2002_optimal_sequence}, with respect to the Hamiltonian of the whole system defined as
$
\mathcal{P}(\rho)=\lim_{T\to\infty}\int_{-T}^T dt e^{-iHt}\rho e^{iHt}/(2T)=\sum_{E_i}\Pi_{E_i}\rho\Pi_{E_i},
$
where $\Pi_{E_i}$ is the projector onto the eigenspace of the Hamiltonian of the whole system corresponding to the eigenvalue $E_i$. 
We then define the pinched black box as
$
 \mP(\mS)\coloneqq \qty{\mP(\rho) ~|~ \rho\in\mS}.
$

The following result shows that the black box work extraction with the time-tranlation covariant condition can be characterized by the composite hypothesis divergence for a pinched black box. (Proof in Appendix B.2.) 
\begin{thm}\label{theorem: one-shot GPC}
Let $H$ and $\tau$ be the Hamiltonian and the thermal state of the system. One-shot extractable work from an arbitrary black box $\mS$ under Gibbs-preserving covariant operations (GPC) satisfies
    \bal\label{Eq: p-one shot extractable work under GPC}
    \beta W^\ve_{\rm{GPC}}(\mS)&=D^\ve_H(\mP(\mS)||\tau).
    \eal
\end{thm}

We would also like to understand the asymptotic limit as Theorem~\ref{theorem: asymptotic GPO}.
However, the structure of the composite hypothesis is more involved in this case because of the correlation between different subsystems generated by the pinching channel. 
This prevents us from directly applying the prior results on composite quantum Stein's lemma~\cite{Brandao_Adversarial_Hypothesis,Berta_2021_On_Composite_HT,Bergh2023composite}.
Indeed, when correlation is present in a composite hypothesis, Stein's lemma can become extremely difficult to handle~\cite{Brandao2010_A_generalization_of,Berta2022_on_a_gap}.
Nevertheless, we show that quantum Stein's lemma holds in our setting. (Proof in Appendix C.2.)
\begin{lem}\label{lemma: pinched composite quantum Stein's lemma}
For an arbitrary sequence $\qty{\mS^{\rm TP}_n(S)}_{n=1}^\infty$ of the sets , it holds that
    \begin{equation}\begin{aligned}
    \lim_{\ve\to +0}\lim_{n\to \infty}\frac{1}{n}D^\ve_H(\mP(\mS^{\rm TP}_n(S))||\tau^{\otimes n})=\min_{\rho\in\mC(S)}D(\rho||\tau).
    \end{aligned}\end{equation}
\end{lem}

We remark that non-composite version of this was previously shown in Ref.~\cite{Lipka_Bartosik_Quantum_dichotomies}.

Let us remark on the relation between Lemma~\ref{lemma: pinched composite quantum Stein's lemma} and the so-called ``generalized quantum Stein's lemma''~\cite{Brandao2010_A_generalization_of}. 
Recent studies have revealed that the relation between (state-aware) resource distillation and hypothesis testing holds at the high level of generality~\cite{Brandao_2008_entanglement_theory,Brandao_2010_reversible_theory,Liu_2019_one-shot,Regula_2021_one-shot}, and both are characterized by the composite hypothesis testing divergence, where the \emph{second} argument of the divergence is a composite hypothesis (while our results contain composite hypothesis in the \emph{first} argument).
The major open question along this line is whether the quantum Stein's lemma holds in this case~\cite{Berta2022_on_a_gap}, whose difficulty rests on the fact that the family of composite hypotheses in the second argument generally has a correlation between different subsystems. 
In this sense, Lemma~\ref{lemma: pinched composite quantum Stein's lemma}, which involves correlation in the composite hypothesis, might be found useful in this context, although it does not appear to directly contribute to the resolution of the problem at the moment.

In the setting of black box work extraction, Lemma~\ref{lemma: pinched composite quantum Stein's lemma} is precisely the one that brings one-shot result (Theorem~\ref{theorem: one-shot GPC}) to the asymptotic setting, which is characterized as follows. 

\begin{thm}\label{theorem: asymptotic GPC}
Let $H$ and $\tau$ be the Hamiltonian and the thermal state of the system. The asymptotic black box extractable work of the sequence of the black boxes $\qty{\mS^{\rm TP}_n(S)}_{n=1}^\infty$ under Gibbs-preserving covariant operations (GPC) is given by
\bal\label{Eq: p-asymptotic extractable work under GPC}
\beta W_{\rm{GPC}}(\qty{\mS^{\rm TP}_n(S)}_{n=1}^\infty)&=\min_{\rho\in\mC(S)}D(\rho||\tau).
\eal
\end{thm}

Theorem~\ref{theorem: asymptotic GPC} shows that although Gibbs-preserving covariant operations come with restrictions compared to Gibbs-preserving operations in one-shot level (as can be seen in Theorems~\ref{theorem: one-shot GPO}~and~\ref{theorem: one-shot GPC}), their performance coincides in the asymptotic limit, both of which are characterized by the standard free energy. 
This result, therefore, extends the similar observation in the standard state-aware work extraction~\cite{Gour_2022_Role_of_quantum_coherence}, in which the work extraction rate also agrees in the asymptotic limit.


\paragraph{Asymptotic black box work extraction under thermal operations.---}
Since the classes of operations considered above are axiomatic, they do not always reflect the physical implementability~\cite{tajima_2024_gibbspreserving}. 
This motivates us to study thermal operations, which is an operationally well-motivated class of thermodynamic processes~\cite{horodecki2013fundamental}. 
Here, we focus on i.i.d. black boxes of the form $\mS^{\rm{i.i.d.}}_n(S):=\qty{\rho^{\otimes n}~|~\rho\in S}$ generated by a set $S$ of finite size, i.e., $\abs{S}<\infty$.
In the state-aware setting, the work extraction from i.i.d. state is discussed in Ref.~\cite{Brandao_2013_RT_of_quantum_states_out_of}, which constructed a protocol that extracts work whose rate asymptotically converges to $D(\rho||\tau)$, where $\rho$ is the known initial state. 

Toward characterizing the asymptotic black box work extraction with thermal operations, we first introduce a new class of thermodynamic processes, which contains thermal operations.  

\begin{defn}
    Let $\mH_A,\mH_B,\mH_C$ be Hilbert spaces, and $\mE:\mD(\mH_A\otimes \mH_B)\to \mD(\mH_C)$ be a CPTP map. We call $\mE$ an \emph{incoherently conditioned thermal operation} if $\mE$ has the form
    $
    \mE=\sum_a\mE^{\rm TO}_a\circ \Lambda^{\rm meas}_a.
    $
    Here, each $\mE_a^{\rm TO}:\mD(\mH_B)\to\mD(\mH_C)$ is a thermal operation and 
    $
    \Lambda^{\rm meas}_a(\rho_{AB}):=\Tr_A\qty[(M^{\rm incoh}_a\otimes I_{B})\rho_{AB}]
    $
    is an instrument representing an incoherent measurement, where $M^{\rm incoh}_a$ is a POVM element satisfying $\mP(M^{\rm incoh}_a) = M^{\rm incoh}_a$.
\end{defn}

We remark that this class clearly contains thermal operations, while this is a subset of the class called conditioned thermal operations introduced in Ref.~\cite{Narasimhachar2017resource}, in which all measurements are allowed to be performed.

First, we show that incoherently conditioned thermal operations perform as well as the two classes of operations discussed above in the asymptotic setting.
The protocol that achieves this is to first learn the given state using some copies and run the state-aware protocol by Ref.~\cite{Brandao_2013_RT_of_quantum_states_out_of}.
In Appendix E.1, we show that, despite the limitation of the available measurement, this protocol is indeed possible by focusing on the structure of pinched states $\mP(\rho^{\otimes n})$. The property of the pinching channel is also referred to in Appendix D.

We now extend this result to thermal operations. 
To this end, we show that in our situation, incoherently conditioned thermal operations coincide with the thermal operations, and we obtain the following result (Proof in Appendix E.2).

\begin{thm}\label{thm:asymptotic TO}
Let $H$ and $\tau$ be the Hamiltonian and the thermal state of the system, respectively.
    The asymptotic black box extractable work of $\qty{\mS^{\rm{i.i.d.}}_n(S)}_{n=1}^\infty$ satisfying $\abs{S}<\infty$ under thermal operations (TO) is represented as
    \bal
    \beta W_{\rm{TO} }\qty(\qty{\mS^{\rm{i.i.d.}}(S)}_{n=1}^\infty)=\min_{\rho\in S}D(\rho||\tau).
    \eal
\end{thm}
In Ref.~\cite{Gour_2022_Role_of_quantum_coherence}, it was shown that the extractable work of the known i.i.d. state is equal to the quantum relative entropy under any of the three free operations mentioned in the discussion above. Our result indicates that the same holds in this case. Whether this holds in the more general setting is not known. We leave a further investigation along this line as a future work.

\paragraph{Black box resource distillation in a general resource theory.---}
As previously discussed, the work extraction task is one example of resource distillation tasks in general resource theories. In Appendix F, we introduce the black box in the general scenario and show that this one-shot distillable resource can also be characterized as the composite hypothesis testing divergence, correspondent to Theorem~\ref{theorem: one-shot GPO}. This indicates the general relation between the state-agnostic resource distillation task and the composite hypothesis testing.

\paragraph{Conclusion.---}
We introduced a framework of black box work extraction, which represents the scenarios where one is to extract work from an unknown quantum state. We presented the optimal guaranteed extractable work in various settings by establishing the connection between one-shot black box work extraction and composite hypothesis testing. We utilized this general relation to characterize the asymptotic work distillation rate by employing and extending quantum Stein's lemma for composite hypothesis testing. Additionally, we devised an explicit protocol for asymptotic black box work extraction for physically motivated classes of thermodynamic processes, which is shown to perform as well as much larger classes of operations. 
We also extended the framework to the general resource theory and clarified the universal relation between the composite hypothesis testing and state-agnostic resource distillation tasks.

Our work clarifies when and how the lack of information about the initial state crucially affects the work extraction performance and what one can still do under such restricted scenarios.
The state-agnostic setting discussed in this work has not been investigated well despite its operational significance and still has much room to explore. 
Potential future directions include an extension of our results to a more general family of black boxes.
Another important extension is to other quantum resource theories beyond one-shot distillation with the maximal set of free operations discussed in this work. 
As our framework forms a new connection between the resource distillation tasks and the quantities in the composite quantum hypothesis testing, the black box resource distillation offers a richer landscape in general quantum resource theories, complementing and extending the state-aware asymptotic distillation tied to generalized quantum Stein's lemma.

\paragraph{Acknowledgments.---} 
We thank Gilad Gour, Bartosz Regula, Christoph Hirsch, and Seth Lloyd for helpful discussions.
We are also grateful to Roberto Rubboli for bringing our attention to the additivity of composite hypothesis testing relative entropy and for sharing his alternative proof with us. 
This work is supported by JSPS KAKENHI Grant Number JP23K19028, JP24K16975, JST, CREST Grant Number JPMJCR23I3, Japan, MEXT KAKENHI Grant-in-Aid for Transformative
Research Areas A ``Extreme Universe” Grant Number JP24H00943, and the
World-Leading Innovative Graduate Study Program for Advanced Basic Science Course at the
University of Tokyo.

\let\oldaddcontentsline\addcontentsline
\renewcommand{\addcontentsline}[3]{}

\bibliographystyle{apsrmp4-2}
\bibliography{myref}

\let\addcontentsline\oldaddcontentsline


\clearpage
\newgeometry{hmargin=1.2in,vmargin=0.8in}

\widetext

\appendix

\setcounter{thm}{0}
\renewcommand{\thethm}{S.\arabic{thm}}
\setcounter{figure}{0}
\renewcommand{\thefigure}{S.\arabic{figure}}

\begin{center}
{\large \bf Appendices}
\end{center}

\tableofcontents

\section{Justification of the definition of the extractable work}\label{app:justification}
The definition of extractable work based on Ref.~\cite{Gour_2022_Role_of_quantum_coherence} employed in the main text is slightly different from the seminal papers discussing the work extraction tasks. 
Here, we show that the definition is equivalent to the standard one based on Ref.~\cite{horodecki2013fundamental}.

In Ref.~\cite{horodecki2013fundamental}, the amount of work is expressed as the energy gaps between the two energy eigenvalues of the Hamiltonian of the two-dimensional system called work storage $W$. We denote the ground state of the work storage as $\ketbra{0}{0}_W$, and excited state as $\ketbra{W}{W}_W$. Also, let $E_W$ be the energy gap between these two levels.
Starting from the initial state $\rho$, if the conversion $\rho\otimes \ketbra{0}{0}_W\to \ketbra{W}{W}_W$ is possible, we say that we can extract the work $E_W$ from initial state $\rho.$ On the other hand, if the conversion $\ketbra{W}{W}_W\to \ketbra{0}{0}_W\otimes \rho$ is possible, we can interpret that the work $E_W$ is sufficient to form the state $\rho$.

On the other hand, in the main text, we consider transforming the thermal state of the battery system $\mu_m=1/m\ketbra{1}{1}_X+(m-1)/m\ketbra{0}{0}_X$ to the charged state $\ketbra{1}{1}_X$.
We now see that being able to accomplish such a transformation is equivalent to the capability of extracting work $(\log m)/\beta$ in the sense of Ref.~\cite{horodecki2013fundamental}.

The results in Ref.~\cite{horodecki2013fundamental} ensure that 
the necessary amount of work $E_W$ to obtain $\ketbra{1}{1}_X$ from the thermal state $\mu_m$ with thermal operations is
\bal
\beta W_{\rm formation}(\ketbra{1}{1}_X)= D_{\rm max}(\ketbra{1}{1}_X||\mu_m)=\log m,
\eal
where $D_{\rm max}$ is the max-divergence defined as 
\bal
D_{\rm max}(\rho||\sigma)=\log\min\qty{\lambda~|~\rho\leq \lambda \sigma}.
\eal
On the other hand, the work $E_W$ extracted from $\ketbra{1}{1}_X$ with thermal operations is
\bal
\beta W_{\rm extractable}(\dm{1}_X)=D_{\rm min}(\ketbra{1}{1}_X||\mu_m)=\log m,
\eal
where $D_{\rm min}$ is the min-divergence defined as 
\bal
D_{\rm min}(\rho||\sigma)=-\log\Tr[\Pi_{{\rm supp}(\rho)}\sigma]~~(\mbox{$\Pi_{{\rm supp}(\rho)}$ is the projector onto ${{\rm supp}(\rho)}$.}).
\eal

These equalities also hold when one can perform the Gibbs-preserving operations~\cite{Faist_2018_Fundamental_work_cost}. Since the work necessary for the formation of the charged state and extractable work from the charged system is the same, we can see that transforming $\mu_m$ to $\dm{1}_X$ is equivalent to obtaining $\dm{W}_W$ with $E_W=(\log m)/\beta$ from $\dm{0}$, i.e., extracting work $(\log m)/\beta$. 

Due to this property, we can consider the work extraction task with the battery system without a loss of generality.

\section{One-shot black box work extraction}
\subsection{One-shot black box work extraction under Gibbs-preserving operations~~(Proof of~Theorem~\ref{theorem: one-shot GPO})}\label{app:one shot extractable work GPO}
\begin{thm}[Theorem~\ref{theorem: one-shot GPO} in the main text]\label{app theorem: one shot extractable work under GPO}
    Let $\mS\subset\mD(\mH)$ be the black box. The one-shot extractable work under the Gibbs-preserving operations is represented as 
    \bal
    \beta W^\ve_{\rm{GPO}}(\mS)=D^\ve_H(\mS||\tau).
    \eal
\end{thm}
\begin{proof}
    We first show the achievable part $\beta W_{\rm{GPO}}(\mS)\geq D^\ve_H(\mS||\tau)$. To show this, it suffices to show that some Gibbs-preserving operation achieves this extractable work yield.
    Consider the CPTP map $\mE:\mD(\mH)\to\mD(\mH_X)$ which has the following form.
    \bal
    \mE(\rho)=\Tr[M\rho]\ketbra{1}{1}_X+\Tr[(I-M)\rho]\ketbra{0}{0}_X, ~~0\leq M\leq I
    \eal
    This map is Gibbs-preserving if and only if this map satisfies $\mE(\tau)=\mu_m$, i.e.,
    \bal
    \mE(\tau)&=\Tr[M\tau]\ketbra{1}{1}_X+\Tr[(I-M)\tau]\ketbra{0}{0}_X=\frac{1}{m}\ketbra{1}{1}_X+\frac{m-1}{m}\ketbra{0}{0}_X,\\
    &\Leftrightarrow m=(\Tr\qty[M\tau])^{-1}.
    \eal
    If we take $M$ such that $M$ satisfies $\Tr\qty[M\rho]\geq 1-\ve $ for every $\rho\in\mS$, we can see that for any $\rho\in \mS$
    \bal
    F(\mE(\rho),\ketbra{1}{1}_X)&=F(\Tr[M\rho]\ketbra{1}{1}_X+\Tr[(I-M)\rho]\ketbra{0}{0}_X,\ketbra{1}{1}_X)\\
    &\geq \Tr[M\rho]F(\ketbra{1}{1}_X,\ketbra{1}{1}_X)+\Tr\qty[(I-M)\rho]F(\ketbra{0}{0}_X,\ketbra{1}{1}_X)\\
    &\geq \Tr[M\rho]\geq 1-\ve.
    \eal
    In the second line, we used the concavity of the fidelity.
    Recalling the definition, the one-shot black box extractable work is calculated as 
    \bal
    \beta W^{\ve}_{\rm{GPO}}(\mS)&=\log\max\qty{m\in\mbR~|~\max_{\mE\in\mbO}\min_{\rho\in\mS}F(\mE(\rho),\ketbra{1}{1}_X)\geq 1-\ve}\\
    &\geq \log \max \qty{\qty(\Tr\qty[M\tau])^{-1}~|~\forall\rho\in\mS,~\Tr\qty[M\rho]\geq 1-\ve,~0\leq M\leq I}.
    \eal
    The last line is the composite hypothesis testing divergence $D^\ve_H(\mS||\tau)$. Therefore, we obtain $\beta W^{\ve}_{\rm{GPO}}(\mS)\geq D^\ve_H(\mS||\tau)$.

    To show the converse part $\beta W_{\rm{GPO}}(\mS)\leq D^\ve_H(\mS||\tau)$, we start by showing the data processing inequality of the composite hypothesis testing divergence, i.e., for any CPTP map $\mE:\mD(\mH)\to\mD(\mH')$ and any composite hypotheses $\mS,\mT\subset\mD(\mH)$,
    \bal
    D^\ve_H(\mE(\mS)||\mE(\mT))\leq D^\ve_H(\mS||\mT)
    \eal
    holds. Recalling the definition of the composite hypothesis testing divergence, $D^\ve_H(\mE(\mS)||\mE(\mT))$ can be written as 
    \bal
    D^\ve_H(\mE(\mS)||\mE(\mT))=-\log \inf_{\substack{0\leq M'\leq I_{\mH'}\\ \sup_{\rho\in \mS}\Tr[(I-M')\mE(\rho)]\leq \ve}}\sup_{\tau\in\mT}\Tr[\mE(\tau) M'].
    \eal
    Here, we denote the conjugate of $\mE$ as $\mE^\dagger$, i.e., $\mE^\dagger$ satisfies $\Tr[\mE(A)B]=\Tr[A\mE^\dagger(B)]$ for any matrices $A\in\mL(\mH_A)$ and $B\in\mL(\mH_B)$.
    Since $\mE$ is a CPTP map, $\mE^\dagger$ is a CP unital map, which maps $I_{\mH'}$ to $I_{\mH}$. This can be rewritten as 
    \bal
    -\log \inf_{\substack{0\leq M'\leq I_{\mH'}\\ \sup_{\rho\in \mS}\Tr[(I-M')\mE(\rho)]\leq \ve}}\sup_{\tau\in\mT}\Tr[\mE(\tau) M']=-\log \inf_{\substack{0\leq M'\leq I_{\mH'}\\ \sup_{\rho\in \mS}\Tr[(I-\mE^\dagger(M'))\rho]\leq \ve}}\sup_{\tau\in\mT}\Tr[\tau \mE^\dagger(M')].
    \eal
    Here, we used the definition and the unitality of $\mE^\dagger.$
    Here, we can easily check that $\mE^\dagger(M')$ satisfies $0\leq \mE^\dagger(M')\leq I_{\mH}$, which follows from the completely positivity of $\mE^\dagger$.
    From this, the following holds.
    \bal
    -\log \inf_{\substack{0\leq M'\leq I_{\mH'}\\ \sup_{\rho\in \mS}\Tr[(I-\mE^\dagger(M'))\rho]\leq \ve}}\sup_{\tau\in\mT}\Tr[\tau \mE^\dagger(M')]\leq -\log \inf_{\substack{0\leq M\leq I_{\mH}\\ \sup_{\rho\in \mS}\Tr[(I-M)\rho]\leq \ve}}\sup_{\tau\in\mT}\Tr[\tau M]= D^\ve_H(\mS||\mT).
    \eal
    
Combining these, we obtain the data processing inequality of the composite hypothesis testing divergence.

If we take the composite alternative hypothesis $\mT$ as $\mT=\qty{\tau}$, the situation is reduced to our original setting.
Let $\mE^*$ be the Gibbs-preserving operation which achieves the optimal work extraction, and $m^* =2^{\beta W^\ve_{\rm{GPO}}(\mS)}$ be the optimal $m$.
From the data processing inequality of the composite hypothesis testing divergence, 
\bal
D^\ve_H(\mS||\tau)&\geq D^\ve_H(\mE^*(\mS)||\mu_{m^*})\\
&=-\log \inf_{\substack{0\leq M\leq I\\ \sup_{\rho\in \mS}\Tr[(I-M)\mE^*(\rho)]\leq \ve}}\Tr[\mu_{m^*}M]
\eal
holds. Recalling that $\mE^*$ satisfies $F(\ketbra{1}{1}_X,\mE^*(\rho))=\Tr[\ketbra{1}{1}_X\mE^*(\rho)]\geq 1-\ve,~~\forall \rho\in \mS$ because of the definition of the extractable work, we can substitute $M=\ketbra{1}{1}_X$ and obtain the following.
\bal
-\log \inf_{\substack{0\leq M\leq I\\ \sup_{\rho\in \mS}\Tr[(I-M)\mE^*(\rho)]\leq \ve}}\Tr[\mu_m M]\geq -\log\Tr\qty[\mu_{m^*}\ketbra{1}{1}_X]=\log m^*=\beta W^\ve_{\rm{GPO}}(\mS).
\eal
From these, we obtain the converse part.
\end{proof}

Since the one-shot extractable work of the state $\rho$ is obtained as $D^\ve_H(\rho||\tau)$(\cite{Gour_2022_Role_of_quantum_coherence}), we can show the direct part in a different way.

\begin{proof}
    (Alternative proof for the direct part.)
    As one can see, if one takes the convex hull on the black box, the conversion fidelity decreases, i.e.,
    \bal
    \max_{\mE\in\rm{GPO}}\min_{\rho\in\mS}F(\mE(\rho),(\ketbra{1}{1}_X,\mu_m))\geq \max_{\mE\in\rm{GPO}}\min_{\rho\in\mC(\mS)}F(\mE(\rho),(\ketbra{1}{1}_X,\mu_m)).
    \eal
    Here, we denote the state of the battery system together with the thermal state of the battery system $\mu_m$.
    From this, one can see that 
    \bal
    \beta W^{\ve}_{\rm{GPO}}(\mS)\geq \beta W^{\ve}_{\rm{GPO}}(\mC(\mS)).
    \eal
    In the following discussion, we focus on the right-hand side. $F(\mE(\rho),\ketbra{1}{1}_X)=\Tr\qty[\mE(\rho)\ketbra{1}{1}_X]$ is linear with respect to the $\mE$ and $\rho$. What is more, the set of the Gibbs-preserving operations and $\mC(\mS)$ are both convex. Since $\mS$ is closed, the convex hull $\mC(\mS)$ is closed. Also, we can see that $\mC(\mS)$ is bounded. From these, we can apply Sion's minimax theorem~\cite{sion_1958_general_minimax_theorems} as follows.
    \bal
    \max_{\mE\in\rm{GPO}}\min_{\rho\in\mC(\mS)}F(\mE(\rho),\ketbra{1}{1}_X)=\min_{\rho\in\mC(\mS)}\max_{\mE\in\rm{GPO}}F(\mE(\rho),\ketbra{1}{1}_X)
    \eal
    Therefore, the $\beta W^{\ve}_{\rm{GPO}}(\mC(\mS))$ is reduced to the following expression.
    \bal
    \beta W^{\ve}_{\rm{GPO}}(\mC(\mS))&=\log\max\qty{m\in\mbR~|~\min_{\rho\in\mC(\mS)}\max_{\mE\in\rm{GPO}}F(\mE(\rho),\ketbra{1}{1}_X)\geq 1-\ve}\\
    &=\min_{\rho\in\mC(\mS)}\beta W^{\ve}_{\rm{GPO}}(\rho)
    \eal
    From the result in Ref.~\cite{Gour_2022_Role_of_quantum_coherence,Wang_RT_of_assymmetric_distinguishability}, $\beta W^{\ve}_{\rm{GPO}}(\rho)=D^\ve_H(\rho||\tau)$, and the above can be rewritten as 
    \bal
    \beta W^{\ve}_{\rm{GPO}}(\mC(\mS))=\min_{\rho\in\mC(\mS)}D^\ve_H(\rho||\tau)\geq D^\ve_H(\mC(\mS)||\tau).
    \eal
    The last inequality follows due to Ref.~\cite[Lemma 16]{Bergh2023composite}. By the definition, the composite hypothesis testing divergence does not change when the convex hull is removed. Combining these discussions, we obtain
    \bal
    \beta W^{\ve}_{\rm{GPO}}(\mS)\geq \beta W^{\ve}_{\rm{GPO}}(\mC(\mS))\geq D^\ve_H(\mC(\mS)||\tau)= D^\ve_H(\mS||\tau).
    \eal
\end{proof}
While Sion's minimax theorem is a very powerful tool in our setting, it is not the appropriate tool in the subsequent discussion in which we discuss the asymptotic limit of the extractable work from the black box.

\subsection{One-shot black box work extraction under Gibbs-preserving covariant operations~(Proof of~Theorem~\ref{theorem: one-shot GPC})}\label{app: one-shot GPC}
To obtain the one-shot black box extractable work under Gibbs-preserving covariant operations, we briefly see the following lemma.
\begin{lem}\label{Lemma: pinching and GPC can be commuted}
    Let $\mE:\mD(\mH_A)\to\mD(\mH_B)$ be a time-translation covariant operation, i.e., $\mE$ satisfies
    \bal\label{Eq: the condition of the time-translation covariance}
    \mE(e^{-iH_A t}\rho e^{iH_At})=e^{-iH_B t}\mE(\rho)e^{iH_B t},~~\forall t\in\mbR,
    \eal
    where $H_A,H_B$ are the Hamiltonians of the input system $A$ and the output system $B$ respectively. 
    Then,
    \bal
    \mP_B\circ\mE=\mE\circ\mP_A
    \eal
    holds, where $\mP_A, \mP_B $ are the pinching channels with respect to the Hamiltonian $H_A$ and $H_B$ respectively.
\end{lem}
\begin{proof}
    If we integrate Eq.~\ref{Eq: the condition of the time-translation covariance}, we obtain
    \bal
    \frac{1}{2T}\int_{-T}^T\dd t~\mE(e^{-iH_A t}\rho e^{iH_At})=\frac{1}{2T}\int_{-T}^T\dd t ~e^{-iH_B t}\mE(\rho)e^{iH_B t}.
    \eal
    Due to the definition of the pinching channel, we obtain Lemma~\ref{Lemma: pinching and GPC can be commuted} by taking the limit $T\to \infty$.
\end{proof}

Furthermore, we define a quantity called conversion fidelity, which represents the optimal guaranteed fidelity of the work extraction task from a given black box achieved by the allowed operations.
\begin{defn}[Ref.~\cite{Gour_2022_Role_of_quantum_coherence}]
    The conversion fidelity of the black box $\mS\subset\mD(\mH)$ under the free operation $\mbO$ is defined as
    \bal
    F_{\mbO}\qty((\mS,\tau)\rightarrow(\ketbra{1}{1}_X,\mu_m))=\max_{\mE\in\mbO}\min_{\rho\in\mS}F(\mE(\rho),\ketbra{1}{1}_X).
    \eal
    Here, we denote the thermal state of the input and battery systems together with the black box of the input states and the output state.
\end{defn}

From Lemma~\ref{Lemma: pinching and GPC can be commuted}, we can see that the conversion fidelity under the Gibbs-preserving covariant operations coincides with that of pinched black box under Gibbs-preserving operations.
\begin{lem}\label{Lemma: Conversion fidelity of GPC and GPO}
    \begin{align}
        F_{\mathrm{GPC}}\qty((\mS,\tau)\rightarrow(\ketbra{1}{1}_X,\mu_m))=F_{\mathrm{GPO}}\qty((\mP(\mS),\tau)\rightarrow(\ketbra{1}{1}_X,\mu_m)),
    \end{align}
    where $\mP(\mS)$ is the set of pinched states of $\mS$
    \begin{equation}
        \begin{aligned}
            \mP(\mS):=\lset \mP(\rho) \sbar \rho\in\mS \rset,
        \end{aligned}
    \end{equation}
    and $\mP$ is a pinching map with respect to the Hamiltonian of the whole system.  
\end{lem}
\begin{proof}
    The idea of the proof is taken from Ref.~\cite{Gour_2022_Role_of_quantum_coherence}.
    First, we show the $(\leq)$ inequality.
    From the definition of the conversion fidelity, 
    \bal
    F_{\mathrm{GPC}}\qty((\mS,\tau)\rightarrow(\ketbra{1}{1}_X,\mu_m))=\max_{\mE\in\rm{GPC}}\min_{\rho\in\mS}F(\mE(\rho),\ketbra{1}{1}_X)
    \eal
    holds. Due to the contractility of the fidelity,
    \bal
    \max_{\mE\in\rm{GPC}}\min_{\rho\in\mS}F(\mE(\rho),\ketbra{1}{1}_X)&\leq\max_{\mE\in\rm{GPC}}\min_{\rho\in\mS}F(\mP\circ\mE(\rho),\mP(\ketbra{1}{1}_X))\\
    &=\max_{\mE\in\rm{GPC}}\min_{\rho\in\mS}F(\mP\circ\mE(\rho),\ketbra{1}{1}_X)
    \eal
    holds. Here, from Lemma~\ref{Lemma: pinching and GPC can be commuted}, 
    \bal
    \max_{\mE\in\rm{GPC}}\min_{\rho\in\mS}F(\mP\circ\mE(\rho),\ketbra{1}{1}_X)=\max_{\mE\in\rm{GPC}}\min_{\rho\in\mS}F(\mE\circ\mP(\rho),\ketbra{1}{1}_X).
    \eal
    Finally, noting that Gibbs-preserving covariant operations are Gibbs-preserving,
    \bal
    \max_{\mE\in\rm{GPC}}\min_{\rho\in\mS}F(\mE\circ\mP(\rho),\ketbra{1}{1}_X)&\leq \max_{\mE\in\rm{GPO}}\min_{\rho\in\mS}F(\mE\circ\mP(\rho),\ketbra{1}{1}_X)\\
    &=F_{\mathrm{GPO}}\qty((\mP(\mS),\tau)\rightarrow(\ketbra{1}{1}_X,\mu_m))
    \eal
    holds. From these, $F_{\mathrm{GPC}}\qty((\mS,\tau)\rightarrow(\ketbra{1}{1}_X,\mu_m))\leq F_{\mathrm{GPO}}\qty((\mP(\mS),\tau)\rightarrow(\ketbra{1}{1}_X,\mu_m))$ is shown.

    Next, we show the opposite inequality. From the definition, we can see
    \bal
    F_{\mathrm{GPO}}\qty((\mP(\mS),\tau)\rightarrow(\ketbra{1}{1}_X,\mu_m))&=\max_{\mE\in\rm{GPO}}\min_{\rho\in\mS}F(\mE\circ\mP(\rho),\ketbra{1}{1}_X)\\
    &\leq \max_{\mE\in\rm{GPO}}\min_{\rho\in\mS}F(\mP\circ\mE\circ\mP(\rho),\ketbra{1}{1}_X).
    \eal
    Again, we used the contractility of fidelity. As one can check easily, $\mP\circ\mE\circ\mP$ is Gibbs preserving covariant for any Gibbs-preserving operation $\mE$. From this,
    \bal
    \max_{\mE\in\rm{GPO}}\min_{\rho\in\mS}F(\mP\circ\mE\circ\mP(\rho),\ketbra{1}{1}_X)\leq \max_{\mE\in\rm{GPC}}\min_{\rho\in\mS}F(\mE(\rho),\ketbra{1}{1}_X)=F_{\mathrm{GPC}}\qty((\mS,\tau)\rightarrow(\ketbra{1}{1}_X,\mu_m))
    \eal
    follows. From these, we $F_{\mathrm{GPC}}\qty((\mS,\tau)\rightarrow(\ketbra{1}{1}_X,\mu_m))\geq F_{\mathrm{GPO}}\qty((\mP(\mS),\tau)\rightarrow(\ketbra{1}{1}_X,\mu_m))$ is shown. Combining the two inequalities, the proof is completed.
\end{proof}
This lemma provides us the expression of the one-shot black box extractable work using the composite hypothesis testing divergence.

\begin{thm}[Theorem~\ref{theorem: one-shot GPC} in the main text]
    Let $\mS\subset\mD(\mH)$ be the black box. The one-shot extractable work under the Gibbs-preserving covariant operations is represented as 
    \bal
    \beta W^\ve_{\rm{GPC}}(\mS)=D^\ve_H(\mP(\mS)||\tau).
    \eal
\end{thm}
\begin{proof}
    From the definition of the one-shot extractable work and Lemma~\ref{Lemma: Conversion fidelity of GPC and GPO},
    \bal
        \beta W^\ve_{\mathrm{GPC}}(\mS)&=\log\max\qty{m\in\mbR~|~F_{\mathrm{GPC}}\qty((\mS,\tau)\rightarrow(\ketbra{1}{1}_X,\mu_m))\geq 1-\ve}\\
        &=\log\max\qty{m\in\mbR~|~F_{\mathrm{GPO}}\qty((\mP(\mS),\tau)\rightarrow(\ketbra{1}{1}_X,\mu_m))\geq 1-\ve}\\
        &=\beta W^\ve_{\mathrm{GPO}}(\mP(\mS))=D^\ve_H(\mP(\mS)||\tau).
    \eal
    holds. We used Lemma~\ref{Lemma: Conversion fidelity of GPC and GPO}. In the last equation, we used the result in Theorem~\ref{app theorem: one shot extractable work under GPO}.
\end{proof}

\section{Asymptotic black box work extraction and composite quantum Stein's lemmas}
\subsection{Asymptotic black box work extraction under Gibbs-preserving operations~(Proof of Theorem~\ref{theorem: asymptotic GPO})}\label{app: asymptotic GPO}
In this section, we consider the asymptotic black box work extraction. To take the asymptotic limit, we have to consider the sequence of black boxes, i.e., the sequence of the subsets of the density matrices $\qty{\mS_n}_{n=1}^\infty,~\mS_n \subset \mD(\mH^{\otimes n}).$ 
Additionally, we consider a specific sequence of black boxes that satisfies the following:
\begin{enumerate}
    \item For any $n\in\mbN$, $\mS_n$ is closed.
    \item For any $n\in\mbN$, $\mS_n$ is closed under the measurement on any subsystems and conditioning on the measurement outcome.
    \item For any $n\in\mbN$, $\mS_n$ is closed under taking partial trace on any subsystems.
    \item For any $n\in\mbN$, $\mS_n$ is closed under permutation of the subsystems.
\end{enumerate}
When we impose these conditions on the sequence of the black boxes, the problem can be reduced to the classical adversarial hypothesis testing, which leads us to obtain the asymptotic limit~\cite{Bergh2023composite, Brandao_Adversarial_Hypothesis}.

Some examples which satisfy these conditions are the following:
\begin{itemize}
    \item The i.i.d. states black box $\mS_n^{\rm{i.i.d.}}(S)=\qty{\rho^{\otimes n}~|~\rho\in S}$   
    \item The tensor product states black box $\mS_n^{\rm{TP}}(S)=\qty{\bigotimes_{i=1}^n \rho_i ~|~\rho_i\in S, \forall i}$
\end{itemize}
Here, $S\subset\mD(\mH)$ is a closed subset.
When we take $S$ as $S=\qty{\rho}$, it is reduced to the trivial black box.

Furthermore, we assume that the thermal state is represented as $\tau^{\otimes n}$, which means that the Hamiltonian of the $n$ systems are the same and have no correlation. 
Then, employing Theorem~\ref{app theorem: one shot extractable work under GPO}, the one-shot asymptotic extractable work of the black box $\mS_n$ under the Gibbs-preserving operations is represented as 
\bal
\beta W^\ve_{\rm{GPO}}(\mS_n)=D^\ve_H(\mS_n||\tau^{\otimes n}).
\eal
To take the $n\to \infty$ limit, we employ a previous result of the composite hypothesis testing and related quantum Stein's lemma.

\begin{pro}[{\cite[Theorem 5]{Bergh2023composite}}]\label{Pro: Bergh2023 composite HT and stein's lemma}
    Let $\mS=\qty{\mS_n}$ and $\mT=\qty{\mT_n}$ be the sequence of black boxes satisfying the conditions above. Then, 
    \begin{align}
        \lim_{\ve\rightarrow0}\lim_{n\rightarrow\infty}\frac{1}{n}D_H^\ve(\mS_n||\mT_n)=\lim_{n\rightarrow\infty}\frac{1}{n}\min_{\substack{\sigma_n\in\mC(\mS_n)\\ \tau_n\in\mC(\mT_n)}}D(\sigma_n||\tau_n),
    \end{align}
    where $\mC(\cdot)$ denotes the convex hull of the set. Furthermore, when for any $n\in \mbN$,  $\mS_n$ is contained by the subspace of $\mD(\mH)$ the dimension of which is polynomial to $n$, the convex hull on $\mS_n$ can be removed.
\end{pro}
When we take $\mT_n=\qty{\tau^{\otimes n}}$, we immediately obtain the following result.

\begin{thm}[Theorem~\ref{theorem: asymptotic GPO} in the main text]\label{Theorem: supplemental material asymptotic extractable work under GPO}
The asymptotic extractable work of the sequence of black boxes $\qty{\mS_n}_{n=1}^\infty$ which satisfies the conditions above under Gibbs-preserving operations is 
    \bal\label{Eq: asymptotic work in GPO}
    \beta W_{\rm{GPO}}(\qty{\mS_n}_{n=1}^\infty)&=\lim_{\ve\to +0}\lim_{n\to \infty}\frac{1}{n}D^\ve_H(\mS_n||\tau^{\otimes n})
    =\lim_{n\to\infty}\frac{1}{n}\min_{\rho_n\in\mC(\mS_n)}D(\rho_n||\tau^{\otimes n})
    \eal
    where $\mC(\cdot)$ denotes the convex hull of the set. Furthermore, when for any $n\in \mbN$,  $\mS_n$ is contained by the subspace of $\mD(\mH)$ the dimension of which is polynomial to $n$, the convex hull on $\mS_n$ can be removed.

Also, If the sequence of the black box is given as the family with a tensor-product structure 
$
\mS^{\rm TP}_n(S):=\qty{\bigotimes_{i=1}^n\rho_i~|~\rho_i\in S, \forall i}
$
generated by an arbitrary set $S\subset \mD(\mH)$, the last equation of Eq.~\eqref{Eq: asymptotic work in GPO} is reduced to a simpler form as
\bal\label{Eq: regularization is not necessary}
\lim_{n\to\infty}\frac{1}{n}\min_{\rho_n\in\mC(\mS_n(S))}D(\rho_n||\tau^{\otimes n})=\min_{\rho\in \mC(S)}D(\rho||\tau).
\eal
\end{thm}

The last statement follows because the relative entropy of the black box closed under the partial trace on any subsystems with respect to the thermal state $\tau^{\otimes n}$ is additive, namely,
\bal
\min_{\rho_n\in\mC(\mS^{\rm TP}_n(S))}D(\rho_n||\tau^{\otimes n})+\min_{\rho_m\in\mC(\mS^{\rm TP}_m(S))}D(\rho_m||\tau^{\otimes m})=\min_{\rho_{n+m}\in\mC(\mS^{\rm TP}_{n+m}(S))}D(\rho_{n+m}||\tau^{\otimes (n+m)}),
\eal
which is shown in the same way as Ref.~\cite{eisert_2003_remarks}. Alternative proof of this additivity can be given~\cite{Rubboli_private_communication} employing similar techniques in Ref.~\cite{Rubboli_2024_new, Rubboli_2023_mixed_state_additivity}.
From this property, it immediately follows that 
\bal\label{Eq: additivity of RE of Black box}
\lim_{n\to\infty}\frac{1}{n}\min_{\rho_n\in\mC(\mS^{\rm TP}_n(S))}D(\rho_n||\tau^{\otimes n})=\lim_{n\to\infty}\frac{1}{n}\cdot n\min_{\rho\in\mC(S)}D(\rho||\tau)=\min_{\rho\in\mC(S)}D(\rho||\tau).
\eal

Furthermore, the convex hull in the regularization of Eq.~\eqref{Eq: asymptotic work in GPO} can be removed when the black box is the i.i.d. states black box. Then, from the additivity of the relative entropy itself,  
\bal\label{Eq: EW of i.i.d. box under GPO}
\beta W_{\rm{GPO}}(\qty{\mS^{\rm{i.i.d.}}_n(S)}_{n=1}^\infty)&=\lim_{n\to\infty}\frac{1}{n}\min_{\rho^{\otimes n}\in\mS_n}D(\rho^{\otimes n}||\tau^{\otimes n})\\
&=\min_{\rho\in S}D(\rho||\tau)
\eal 
holds. In this case, the asymptotic black box extractable work equals the worst-case extractable work in the normal setting where the experimenters have complete information.

\subsection{Asymptotic black box work extraction under Gibbs-preserving covariant operations~(Proofs of Lemma~\ref{lemma: pinched composite quantum Stein's lemma} and Theorem~\ref{theorem: asymptotic GPC})}\label{app: asymptotic GPC}
In the same way as the previous discussion, for a given sequence of the black boxes $\qty{\mS_n}_{n=1}^\infty$, the asymptotic black box work extraction under Gibbs-preserving covariant operations is expressed as follows.
\bal
\beta W_{\rm{GPC}}(\qty{\mS_n}_{n=1}^\infty)&=\lim_{\ve\to +0}\lim_{n\to \infty}\frac{1}{n}D^\ve_H(\mP(\mS_n)||\tau^{\otimes n})
\label{eq:asymptotic GPC with hypothesis testing}
\eal
The RHS is more complicated than the LHS of Proposition~\ref{Pro: Bergh2023 composite HT and stein's lemma}, since $\mP(\mS_n)$ no longer has the tensor-product structure, which means that $\mP(\mS_n)$ is not closed under the measurement on any subsystems and conditioning on the measurement result. Here, we show that another type of composite quantum Stein's lemma holds even in this case.
\begin{pro}[Lemma~\ref{lemma: pinched composite quantum Stein's lemma} in the main text]\label{proposition app: pinched composite quantum Stein's lemma}
Let $\qty{\mS_n}_{n=1}^\infty$ be a sequence of black boxes which satisfies the conditions above. Here, the following holds.
    \bal
    \lim_{\ve\to +0}\lim_{n\to \infty}\frac{1}{n}D^\ve_H(\mP(\mS_n)||\tau^{\otimes n})=\lim_{n\to \infty}\frac{1}{n}\min_{\rho_n\in \mC(\mS_n)}D(\rho_n||\tau^{\otimes n})
    \eal

Furthermore, when the sequence of the black box is given as the family with a tensor-product structure, Eq.~\eqref{Eq: regularization is not necessary} implies that
\bal
\lim_{\ve\to +0}\lim_{n\to \infty}\frac{1}{n}D^\ve_H(\mP(\mS^{\rm TP}_n(S))||\tau^{\otimes n})=\min_{\rho\in \mC(S)}D(\rho||\tau).
\eal

\end{pro}
Note that a similar type of the quantum Stein's lemma can be seen in Ref.~{\cite[Lemma 4]{Lipka_Bartosik_Quantum_dichotomies}}, and is included by Proposition~\ref{proposition app: pinched composite quantum Stein's lemma}.

To show this, we start from showing the following lemma.
\begin{lem}\label{lemma: relative entropy does not change under pinching in the limit}
    Let $\rho,\tau\in\mD(\mH)$ be arbitrary states, and $\mP$ be the pinching channel with respect to $\tau$. Then, the following holds.
    \bal
    0\leq D(\rho||\tau)-D(\mP(\rho)||\tau)\leq \log\abs{\spec(\tau)},
    \eal
    where $\abs{\spec(\tau)}$ is the number of the different eigenvalues of $\tau$.
\end{lem}
\begin{proof}
    $0\leq D(\rho||\tau)-D(\mP(\rho)||\tau)$ is shown by using the data processing inequality of the relative entropy. To show the other inequality, we first employ Hayashi's pinching inequality~\cite{Masahito_Hayashi_2002_optimal_sequence}
    \bal
    \mP(\rho)\geq\frac{\rho}{\abs{\spec(\tau)}}.
    \eal
    Due to this inequality and the properties of the relative entropy, the following holds.
    \bal
    D(\rho||\tau)&=D(\mP(\rho)||\tau)+D(\rho||\mP(\rho))\\
    &\leq D(\mP(\rho)||\tau)+D(\rho||\rho/\abs{\spec(\tau)})=D(\mP(\rho)||\tau)+\log\abs{\spec(\tau)},
    \eal
    where in the inequality, we used the operator monotonicity of $\log$.
\end{proof}

In the subsequent discussion, we denote the CPTP map which represents the measurement whose POVM elements are $\qty{M_a}_a$ as
\bal
\mM(\rho):=\sum_a \Tr[\rho M_a]\ketbra{a}{a},
\eal
where $\qty{\ket{a}}_a$ is the orthogonal vectors in the classical system.
We review the concepts called compatible pair.
\begin{defn}[Ref.~\cite{Brandao_Adversarial_Hypothesis}]
    Let $\vec{\mbM}=(\mbM_1,\mbM_2,\ldots)$ be the sequence of the measurements with $\mbM_n$ reprsenting the set of measurements on $\mD(\mH^{\otimes n})$. 
    Furthermore, let $\mS=(\mS_1,\mS_2,\ldots)$ be the sequence of the sets of state where $\mS_n$ is the subset of $\mD(\mH^{\otimes n})$ for every $n\in\mbN$.
    We say that $(\vec{\mbM},\mS)$ is the compatible pair when the sequence $\mS$ is closed under the measurement in $\vec{\mbM}$ on any subsystems and conditioning on the measurement outcome, i.e., for any state $\rho_{n+k}\in\mS_{n+k}$, after performing any measurement in $\mbM_k$ and conditioning on the outcome, the post-measured state is the element of $\mS_n$.
\end{defn}

Furthermore, we consider the restricted set of measurements called incoherent measurements, which guarantees that the probability distribution obtained by such a measurement is invariant under time translation.

\begin{defn}
    Let $H$ be the Hamiltonian of the considered system, and $\qty{M_a}_a$ be a set of POVM elements of a measurement $\mM$. We say that $\mM$ is an incoherent measurement if and only if all the POVM elements $M_a$ satisfy the following condition.
    \bal
    \forall a,~~\mP(M_a)=M_a.
    \eal
    We denote the sequence of the incoherent measurement as $\vec{\mbM}^{\rm incoh}$.
\end{defn}
From the definition, the nonzero elements of the POVM elements of the incoherent measurement are in the energy blocks. From this observation, it holds that $M^{\rm incoh}_a=\exp(-iHt)M^{\rm incoh}_a\exp(iHt),~\forall a, ~\forall t\in\mbR$. This property guarantees that the probability distribution of the measurement outcome is invariant under the time translation, i.e.,
\bal
\Tr\qty[e^{-iHt}\rho e^{iHt} M^{\rm incoh}_a]&=\Tr\qty[\rho e^{iHt}M^{\rm incoh}_a e^{-iHt} ]\\
&=\Tr[\rho M^{\rm incoh}_a].
\eal
We remark that this is a special case of the time-translation covariant measurement~\cite{noauthor_probabilistic_nodate}, the measurement whose POVM elements $\qty{M_a^t}_a$  satisfy
\bal
e^{iH\Delta t}M^t_ae^{-iH\Delta t}=M^{t+\Delta t}_a~\forall a,~\forall t,\Delta t\in\mbR.
\eal

We also denote the sequence of all measurements as $\vec{\mbM}^{\rm all}$. Before showing Proposition~\ref{proposition app: pinched composite quantum Stein's lemma}, we show the following two lemmas.
\begin{lem}\label{Lemma: incoherent measurement and pinched black box is the compatible pair}
Let $\mS=\qty{\mS_n}_{n=1}^\infty$ be the sequence of the black boxes satisfying the condition mentioned above. Furthermore, let $\mP(\mC(\mS))$ be the sequence of the sets of states $\qty{\mP(\mC(\mS_n))}_{n=1}^\infty$.
Then, $(\vec{\mbM}^{\rm incoh},\mP(\mC(\mS)))$ is the compatible pair.
\end{lem}
Note that $(\vec{\mbM}^{\rm incoh},\mP(\mS))$ is the compatible pair too, which is proven in the same way as the following proof. Here, we consider the convex hull of the black boxes to use this lemma to show  Proposition~\ref{proposition app: pinched composite quantum Stein's lemma}.
\begin{proof} of Lemma~\ref{Lemma: incoherent measurement and pinched black box is the compatible pair}.
It suffices to show that for any POVM elements $M^{\rm incoh}_a$ and the arbitrary state $\mP(\rho_{n+k})\in \mP(\mC(\mS_{n+k})), \forall n,k\in\mbN$ 
\bal
\Tr_{n+1,\ldots,n+k}\qty[\qty(I\otimes M^{\rm incoh}_a)\mP\qty(\rho_{n+k})]\in\mP(\mC(\mS_{n}))
\eal
holds.
This can be checked as follows.
\begin{equation}\begin{aligned}
&\Tr_{n+1,\ldots,n+k}\qty[\qty(I\otimes M^{\rm incoh}_a)\mP\qty(\rho_{n+k})]\\
&=\Tr_{n+1,\ldots,n+k}\qty[\mP(I\otimes (M^{\rm incoh}_a))\rho_{n+k}]\\
&=\Tr_{n+1,\ldots,n+k}\qty[(I\otimes M^{\rm incoh}_a)\rho_{n+k}]\\
& = \Tr_{n+1,\ldots,n+k}\qty[(\mP_{1\dots n}(I)\otimes M^{\rm incoh}_a)\rho_{n+k}]\\
& = \mP_{1,\dots,n}\left(\Tr_{n+1,\ldots,n+k}\qty[(I\otimes M^{\rm incoh}_a)(\rho_{n+k})]\right)
\in\mP(\mC(\mS_{n})).
\end{aligned}\end{equation}
The first equality is because the pinching channel satisfies $\mP^\dagger=\mP$, and the property mentioned at the beginning of this proof is used in the second equality.
The third equality is because $\mP_{1,\dots,n}(I)=I$, where $\mP_{1,\dots,n}$ is the pinching channel applied to the first $n$ subsystems. 
To show the final inclusion, we used the second property of the sequence of the black boxes $\mS=\qty{\mS_n}_{n=1}^\infty$ mentioned at the beginning of this section.
\end{proof}

\begin{lem}\label{Lemma: incoherent measurement and pinching plus all measurement}
    It holds that
    \bal
\sup_{\mM\in \mbM^{\rm incoh}_n}\min_{\rho_n\in \mC(\mS_n)}D(\mM(\rho_n)||\mM(\tau^{\otimes n}))=\sup_{\mM\in \mbM^{\rm{all}}_n}\min_{\rho_n\in \mC(\mS_n)}D(\mM(\mP(\rho_n))||\mM(\tau^{\otimes n})).
    \eal
\end{lem}
\begin{proof}
    First, we show the $(\leq )$ inequality. For any states $\rho_n$ and any POVM element of the incoherent measurement $M_a^{\rm incoh}$, noting that $\mP(M_a^{\rm incoh})=M_a^{\rm incoh}$, it holds that
\bal
\Tr[\rho_n M^{\rm incoh}_a]=\Tr[\rho_n \mP\qty(M^{\rm incoh}_a)]=\Tr[ \mP(\rho_n) M^{\rm incoh}_a].
\eal
This implies that all the probability distributions that are obtained by measuring the state $\rho_n$ with the incoherent measurement can be realized by measuring the pinched state $\mP(\rho_n)$ with the incoherent measurement. From this, 
\bal
\sup_{\mM\in \mbM^{\rm incoh}_n}\min_{\rho_n\in \mC(\mS_n)}D(\mM(\rho_n)||\mM(\tau^{\otimes n}))\leq\sup_{\mM\in \mbM^{\rm{all}}_n}\min_{\rho_n\in \mC(\mS_n)}D(\mM(\mP(\rho_n))||\mM(\tau^{\otimes n})).
\eal
holds. To show the $(\geq )$ inequality, note that for any POVM elements $\qty{M_a^{\rm all}}_a$ which represents a measurement in $\mbM_n^{\rm all}$, it holds that
\bal
\Tr[\mP(\rho_n) M^{\rm all}_a]=\Tr[\rho_n \mP(M^{\rm all}_a)].
\eal
Here, $\qty{\mP(M^{\rm all}_a)}_a$ satisfies the conditions for POVM elements, i.e.,
\bal
\forall a,~\mP(M^{\rm all}_a)\geq 0, ~~\sum_a\mP(M^{\rm all}_a)=\mP\qty(\sum_a M^{\rm all}_a)=\mP(I)=I
\eal
holds.
Furthermore, due to the definition of the pinching channel, the measurement whose POVM elements are represented as $\qty{\mP(M^{\rm all}_a)}_a$ is the incoherent measurement. From this, one can see that all the probability distribution obtained by measuring the pinched state $\mP(\rho_n)$ with any measurement can be realized by measuring $\rho_n$ with incoherent measurement which implies
\bal
\sup_{\mM\in \mbM^{\rm incoh}_n}\min_{\rho_n\in \mC(\mS_n)}D(\mM(\rho_n)||\mM(\tau^{\otimes n}))\geq\sup_{\mM\in \mbM^{\rm{all}}_n}\min_{\rho_n\in \mC(\mS_n)}D(\mM(\mP(\rho_n))||\mM(\tau^{\otimes n})),
\eal
which completes the proof.
\end{proof}

We are ready to show  Proposition~\ref{proposition app: pinched composite quantum Stein's lemma}.
\begin{proof} of  Proposition~\ref{proposition app: pinched composite quantum Stein's lemma}.
Note that 
\bal
D^\ve_H(\mP(\mS_n)||\tau^{\otimes n})=D^\ve_H(\mC(\mP(\mS_n))||\tau^{\otimes n})
\eal
holds and $D^\ve_H(\mC(\mP(\mS_n))||\tau^{\otimes n})$ can be interpreted as the hypothesis testing divergence of $\mC(\mS_n)$ with respect to $\tau^{\otimes n}$ when the allowed measurement is restricted to the incoherent measurement. Due to Lemma~\ref{Lemma: incoherent measurement and pinched black box is the compatible pair} and \cite[Theorem 16]{Brandao_Adversarial_Hypothesis}, the following holds.
\bal
\lim_{\ve\to +0}\lim_{n\to \infty}\frac{1}{n} D^\ve_H(\mC(\mP(\mS_n))||\tau^{\otimes n})&=\lim_{n\to \infty}\frac{1}{n}\sup_{\mM\in \mbM^{\rm incoh}_n}\min_{\rho_n\in \mC(\mP(\mS_n))}D(\mM(\rho_n)||\mM(\tau^{\otimes n}))\\
&=\lim_{n\to \infty}\frac{1}{n}\sup_{\mM\in \mbM^{\rm incoh}_n}\min_{\rho_n\in \mC(\mS_n)}D(\mM(\rho_n)||\mM(\tau^{\otimes n}))\\
&=\lim_{n\to \infty}\frac{1}{n}\sup_{\mM\in \mbM^{\rm{all}}_n}\min_{\rho_n\in \mC(\mS_n)}D(\mM(\mP(\rho_n))||\mM(\tau^{\otimes n}))
\eal
Here, the second line is because $\mM(\rho_n)=\mM\circ\mP(\rho_n)$ for $\mM\in\mbM^{\rm incoh}_n$ and thus the minimization for $\rho_n$ over $\mC(\mS_n)$ coincides with that over $\mC(\mP(\mS_n))$, and the third line follows from Lemma~\ref{Lemma: incoherent measurement and pinching plus all measurement}.
Employing Ref.~\cite[Lemma 13]{Brandao_Adversarial_Hypothesis}, we can exchange the $\sup$ and $\min$, i.e., 
\bal
\lim_{n\to \infty}\frac{1}{n}\sup_{\mM\in \mbM^{\rm{all}}_n}\min_{\rho_n\in \mC(\mS_n)}D(\mM(\mP(\rho_n))||\mM(\tau^{\otimes n}))&=\lim_{n\to \infty}\frac{1}{n}\min_{\rho_n\in \mC(\mS_n)}\sup_{\mM\in \mbM^{\rm{all}}_n}D(\mM(\mP(\rho^{\otimes n}))||\mM(\tau^{\otimes n}))\\
&=\lim_{n\to\infty}\frac{1}{n}\min_{\rho_n\in\mC(\mS_n)}D_{\mbM_n^{\rm{all}}}(\mP(\rho_n)||\tau^{\otimes n}),
\eal
where $D_{\mbM^{\rm{all}}_n}(\mP(\rho_n)||\tau^{\otimes n})$ is the $\mbM_n^{\rm{all}}$-measured relative entropy of $\mP(\rho_n)$ with respect to $\tau^{\otimes n}$~\cite{Piani_2009_relative_entropy_of}.

We first note that the infimum is achieved at a permutation invariant state~\cite[Lemma 23]{Bergh2023composite}.
Therefore, letting ${\rm PI}_n$ be the set of $n$-qudit permutation invariant states, we get 
\bal
\lim_{n\to \infty}\frac{1}{n}\min_{\rho_n\in \mC(\mS_n)}D_{\mbM_n^{\rm all}}(\mP(\rho_n)||\tau^{\otimes n}) = \lim_{n\to \infty}\frac{1}{n}\min_{\rho_n\in \mC(\mS_n)\cap {\rm PI}_n}D_{\mbM_n^{\rm all}}(\mP(\rho_n)||\tau^{\otimes n}).
\eal
We now recall Ref.~\cite[Lemma 2.4]{Berta_2021_On_Composite_HT}, showing that for all permutation invariant states $\eta_n$ and $\sigma_n$, it holds that
\bal
 D(\eta_n\|\sigma_n) - {\rm log\,poly}(n)\leq D_{\mbM_{\rm all}}(\eta_n\|\sigma_n) \leq D(\eta_n\|\sigma_n).
 \label{eq:measured rel and rel permutation invariant}
\eal
This implies
\bal
\lim_{n\to \infty}\frac{1}{n}\min_{\rho_n\in \mC(\mS_n)\cap {\rm PI}_n}D_{\mbM_{\rm all}}(\mP(\rho_n)||\tau^{\otimes n}) = \lim_{n\to \infty}\frac{1}{n}\min_{\rho_n\in \mC(\mS_n)\cap{\rm PI}_n}D(\mP(\rho_n)||\tau^{\otimes n}).
\eal

Note that the number of the different eigenvalue of $H^{\times n}$ is upper bounded by the number of type classes of $n$ length strings when the set of alphabets is $\qty{0,1,\ldots, d-1}$. Since the number of type classes is upper-bounded by $(n+1)^d$~\cite{Cover_Thomas}, due to Lemma~\ref{lemma: relative entropy does not change under pinching in the limit}, it holds that for any states $\rho_n$,
\bal
 D(\rho_n||\tau^{\otimes n})- {\rm log\,poly}(n)\leq D(\mP(\rho_n)||\tau^{\otimes n})\leq D(\rho_n||\tau^{\otimes n}),
\eal
which implies
\bal
\lim_{n\to \infty}\frac{1}{n}\min_{\rho_n\in \mC(\mS_n)\cap{\rm PI}_n}D(\mP(\rho_n)||\tau^{\otimes n})=\lim_{n\to \infty}\frac{1}{n}\min_{\rho_n\in \mC(\mS_n)\cap{\rm PI}_n}D(\rho_n||\tau^{\otimes n}).
\eal
Again, employing Ref.~\cite[Lemma 23]{Bergh2023composite}, we obtain
\bal
\lim_{n\to \infty}\frac{1}{n}\min_{\rho_n\in \mC(\mS_n)\cap{\rm PI}_n}D(\rho_n||\tau^{\otimes n})=\lim_{n\to \infty}\frac{1}{n}\min_{\rho_n\in \mC(\mS_n)}D(\rho_n||\tau^{\otimes n}).
\eal
Due to Eq.~\eqref{Eq: additivity of RE of Black box}, it holds that
\bal
\lim_{n\to \infty}\frac{1}{n}\min_{\rho_n\in \mC(\mS_n)}D(\rho_n||\tau^{\otimes n})=\min_{\rho\in \mC(S)}D(\rho||\tau).
\eal
Combining these, we complete the proof.

\end{proof}

Combining Proposition~\ref{proposition app: pinched composite quantum Stein's lemma} with \eqref{Eq: regularization is not necessary}, we finally obtain the following result. 

\begin{thm}[Theorem~\ref{theorem: asymptotic GPC} in the main text]\label{Theorem: supplemental material asymptotic extractable work under GPC}
    The asymptotic black box extractable work of the sequence of the black boxes $\qty{\mS_n}_{n=1}^\infty$ under Gibbs-preserving covariant operations is given by
\bal\label{Eq: p-asymptotic extractable work under GPC}
\beta W_{\rm{GPC}}(\qty{\mS_n}_{n=1}^\infty)&=\lim_{n\to \infty}\frac{1}{n}\min_{\rho_n\in \mC(\mS_n)}D(\rho_n||\tau^{\otimes n}).
\eal

Furthermore, when the sequence of the black box is given as the family with a tensor-product structure, Eq.~\eqref{Eq: regularization is not necessary} implies that
\bal
\beta W_{\rm{GPC}}(\qty{\mS^{\rm TP}_n(S)}_{n=1}^\infty)=\min_{\rho\in \mC(S)}D(\rho||\tau).
\eal

\end{thm}

\subsection{Application to some examples}
The results indicate the fundamental difficulty of the work extraction task without information about the initial state. We exhibit some examples that show this difference.

Consider the $n$-qubit system $\mH_2^{\otimes n}$ with $\mH_2=\Span\qty{\ket{0},\ket{1}}$. 
Suppose that the system is equipped with the trivial Hamiltonian, i.e., $H=0$, which comes with the thermal state $I^{\otimes n}/{2^n}$.
Consider the setting where we aim to extract as much work as possible from a given state
\bal\label{Eq: app ideal state}
\ket{\psi}=\frac{4}{5}\ket{0}+\frac{3}{5}\ket{1}.
\eal
By applying a Gibbs preserving operation, the state is exposed to bit-flip noise on some of the subsystems. We compare the two settings---the one with the knowledge of the error location and the one without.
This situation is exactly that of the black box work extraction from the black box 
\bal
\mS_n=\lset\rho^{(n)}=\bigotimes_{i=1}^n\rho_i~\sbar~\rho_i\in\qty{\ketbra{\psi}{\psi},X\ketbra{\psi}{\psi}X}, \forall i\in\qty{1,\ldots, n}\rset.
\eal
We consider the performance of the work extraction task in the asymptotic limit.

We first evaluate the extractable work of the sequence $\qty{\rho^{(n)}}_{n=1}^\infty$ of the element of $\qty{\mS_n}_{n=1}^\infty$ when we have the complete information of this sequence.
The asymptotic extractable work of the sequence $\qty{\rho^{(n)}}_{n=1}^\infty$ is
\bal\label{Eq: supprement work in multi qubit system}
\beta W_{\rm GPO}\qty(\qty{\rho^{(n)}}_{n=1}^\infty)=\lim_{\ve\to +0}\limsup_{n\to \infty}\frac{1}{n}D^\ve_H(\rho^{(n)}||\tau^{\otimes n}),
\eal
which is the consequence from Ref.~\cite{Wang_2019_resource}.
To calculate the lower bound of this quantity, we show the following lemma.
\begin{lem}\label{Lemma: supprement irregular form of superaddditivity of HT divergence}
    Let $\rho_1,\sigma_1\in\mD(\mH_1)$ and $\rho_2,\sigma_2\in\mD(\mH_2)$ be the arbitrary states. Then, the following holds.
    \bal
    D^\ve_H(\rho_1\otimes \rho_2||\sigma_1\otimes \sigma_2)\geq  D^\ve_H(\rho_1||\sigma_1)+D_{\rm min}(\rho_2||\sigma_2)
    \eal
\end{lem}
\begin{proof}
    From the definition of the hypothesis testing divergence and min divergence, the statement can be seen as follows.
    \bal
    D^\ve_H(\rho_1\otimes \rho_2||\sigma_1\otimes \sigma_2)&=-\log\min_{\substack{0\leq M_{12}\leq I_{12} \\ \Tr[\rho_1\otimes \rho_2 M_{12}]\geq 1-\ve }}\Tr[\sigma_1\otimes \sigma_2 M_{12}]\\
    &\geq -\log\min_{\substack{0\leq M_{1}\leq I_{1},0\leq M_2\leq I_2 \\ \Tr[(\rho_1\otimes \rho_2)(M_1\otimes M_2)]\geq 1-\ve }}\Tr[(\sigma_1\otimes \sigma_2) (M_1\otimes M_2)]\\
    &\geq -\log\min_{\substack{0\leq M_{1}\leq I_{1} \\ \Tr[(\rho_1\otimes \rho_2)(M_1\otimes \Pi_{\supp (\rho_2)})]\geq 1-\ve }}\Tr[(\sigma_1\otimes \sigma_2) (M_1\otimes \Pi_{\supp(\rho_2)})]\\
    &=-\log\min_{\substack{0\leq M_{1}\leq I_{1} \\ \Tr[\rho_1 M_1]\geq 1-\ve }}\Tr[\sigma_1 M_1]-\log\Tr[\sigma_2\Pi_{\supp (\rho_2)}]\\
    &=D^\ve_H(\rho_1||\sigma_1)+D_{\rm min}(\rho_2||\sigma_2)
    \eal
\end{proof}
Applying the inequality in Lemma~\ref{Lemma: supprement irregular form of superaddditivity of HT divergence} several times, the hypothesis testing divergence appearing in Eq.~\eqref{Eq: supprement work in multi qubit system} is lower bounded as 
\bal
D^\ve_H(\rho^{(n)}||\tau^{\otimes n})\geq \sum_{i=1}^n D_{\rm min}(\rho_i||\tau).
\eal
Here, we used $D^\ve_H(\rho||\sigma)\geq D_{\rm min}(\rho||\sigma)$.
Since $\tau=I/2$ and $\rho_i$ is either $\ket{\psi}$ or $\ket{\psi'}:=X\ket{\psi}$, when $\rho_i=\ketbra{\psi}{\psi}$, it holds that
\bal
D_{\rm min}(\ketbra{\psi}{\psi}||I/2)=-\log\Tr\qty[\ketbra{\psi}{\psi}\frac{I}{2}]=1,
\eal
and similarly $D_{\rm min}(\ketbra{\psi'}{\psi'}||I/2)=1$ holds.
Combining these, we obtain the lower bound of the extractable work from the sequence in the state-aware scenario as 
\bal
\beta W_{\rm GPO}\qty(\qty{\rho^{(n)}}_{n=1}^\infty)&=\lim_{\ve\to +0}\limsup_{n\to \infty}\frac{1}{n}D^\ve_H(\rho^{(n)}||\tau^{\otimes n})\\
&\geq \lim_{\ve\to +0}\limsup_{n\to \infty}\frac{1}{n}\sum_{i=1}^n D_{\rm min}(\rho_i||I/2)\\
&=\lim_{\ve\to +0}\limsup_{n\to \infty}\frac{1}{n}\sum_{i=1}^n 1=1.
\eal
Therefore, for any sequence $\qty{\rho^{(n)}}_{n=1}^\infty$, the asymptotic extractable work is equal to or larger than 1.

We next consider the asymptotic extractable work from the sequence of the black box $\qty{\mS_n}_{n=1}^\infty$. Applying Theorem~\ref{Theorem: supplemental material asymptotic extractable work under GPO}, the asymptotic extractable work from the sequence of the black box is 
\bal
\beta W_{\rm GPO}(\qty{\mS_n}_{n=1}^\infty)=\min_{\sigma\in\mC(S)}D(\sigma||\tau).
\eal
Notice that $\sigma:=\frac{1}{2}\ketbra{\psi}{\psi}+\frac{1}{2}\ketbra{\psi'}{\psi'}\in\mC(S)$. Since it holds that
\bal
D\qty(\sigma||\frac{I}{2})=1-S(\sigma)=1-h\qty(\frac{1}{50})<1.
\eal
Here, $S(\sigma)\coloneqq -\Tr(\sigma\log\sigma)$ is the von Neumann entropy, and $h(p)$ is the binary entropy.
From these, it holds that
\bal
\beta W_{\rm GPO}(\qty{\mS_n}_{n=1}^\infty)=\min_{\sigma\in\mC(S)}D(\sigma||\tau)\leq D\qty(\sigma||\frac{I}{2})=1-h\qty(\frac{1}{50})<1.
\eal

Comparing the two scenarios, it follows that
\bal
\beta W_{\rm GPO}(\qty{\mS_n}_{n=1}^\infty)\leq 1-S(\sigma)<1\leq \beta W_{\rm GPO}\qty(\qty{\rho^{(n)}}_{n=1}^\infty)
\eal
for any sequence $\qty{\rho^{(n)}}_{n=1}^\infty$. This means that $\qty{\mS_n}_{n=1}^\infty$ is the sequence of the black box whose extractable work is strictly smaller when we do not have complete information about the initial state.

It is also interesting to compare situations in which one extracts work from the given state with knowledge about the probability of the noise to those without knowledge.
To this end, let us consider another example. Suppose that ideally we are given $n$ copies of $\ket{\psi}$ in Eq. ~\eqref{Eq: app ideal state}, and each subsystem experiences the bit-flip noise with a small probability $p$.
We also assume that, as the example above, the Hamiltonian of each system is fully degenerate.
If we know the probability $p$ of the error, for each subsystems, we obtain $\ket{\psi}$ with probability $1-p$ and $\ket{\psi'}$ with probability $p$. Because of this, The given state is
\bal
\rho^{\otimes n}=\qty((1-p)\ketbra{\psi}{\psi}+p\ketbra{\psi'}{\psi'})^{\otimes n}.
\eal
The extractable work from this sequence of states is
\bal
\beta W_{\rm GPO}(\rho^{\otimes n})=\lim_{\ve\to+0}\lim_{n\to \infty}\frac{1}{n}D^\ve_H(\rho^{\otimes n}||(I/2)^{\otimes n})=D(\rho||I/2)=1-S(\rho).
\eal
When $p$ is sufficiently small, it can be shown that $S(\rho)<h\qty(\frac{1}{50})$. 
On the other hand, the situation without knowledge of the error probability corresponds to the black box work extraction setting mentioned above regardless of the probability $p$.
From these, it holds that
\bal
\beta W_{\rm GPO}(\qty{\mS_n}_{n=1}^\infty)\leq 1-h\qty(\frac{1}{50})<1-S(\rho)=\beta W_{\rm GPO}(\rho^{\otimes n}),
\eal
which means that the extractable work without the information about the error probability is strictly smaller than that with the information. This example indicates the fundamental restriction of the work extraction tasks without knowing the probability distribution with which the initial state is given.

A more extreme example that represents the restriction of the work extraction task from a black box is the following.
Consider the family of black boxes with a tensor-product structure generated by the set $S=\qty{\ketbra{\phi_1}{\phi_1},\ldots, \ketbra{\phi_d}{\phi_d}}$, where $\ket{\phi_1},\ldots,\ket{\phi_d}$ are the eigenstates of the Hamiltonian $H$ of a single subsystem.
If we have information about the initial state, we can extract the nonzero work from it since the free energy of any state in the black box is strictly larger than that of the thermal states. 
However, the thermal state $\tau$ is included in $\mC(S)$, which implies that one cannot extract any work from this sequence of the black boxes asymptotically. This example corresponds to the multi-qubit system mentioned above with $\ket{\psi}=\ket{0}$ and $\ket{\psi'}=\ket{1}$.

These examples reveal the stark difference between the standard state-aware work extraction task and the black box work extraction task.
We remark that, as is also mentioned in the main text and discussed in the above example, our black box framework assumes that we do not even have access to the probability distribution according to which the initial states are chosen from the box.
If we know the probability distribution, the situation is no longer the state-agnostic scenario because the initial state is the averaged state taken over the black box with the weight of the probability distribution. The absence of such a probability distribution makes the results nontrivial.

\section{Convexity of thermal operation and pinching channel}\label{app: convexity of thermal operation}

Recalling the definition of pinching channel $\mP(\cdot)=\lim_{T\to\infty}\int_{-T}^T dt e^{-iHt}\cdot e^{iHt}/(2T)$, the fact that the pinching channel is a thermal operation is a direct consequence of the convexity of the set of thermal operations. 
The convexity of the set of thermal operations is proven in Ref.~\cite{Lostagilo_Quantum_coherence_time_translation_symmetry, gour2024resources_of_quantum_world}. In particular, the proof in Ref.~{\cite[Theorem 17.2.1]{gour2024resources_of_quantum_world}} deals with a general case where the input and output systems of the channel can be different. 
Here, for completeness, we provide a proof for the convexity of thermal operations based on the argument in Ref.~\cite{gour2024resources_of_quantum_world}, where we also offer additional insights about the specific construction of a probabilistic mixture of thermal operations as a single thermal operation.

Consider a channel $\mN:\mD(\mH_A)\rightarrow\mD(\mH_B)$ of the form
\bal
\mN=\sum_{x=1}^m p_x\mN_x
\eal
where $\mN_x:\mD(\mH_A)\rightarrow\mD(\mH_B)$ is a thermal operation for every $x\in\qty{1,\cdots, m}$, and $\qty{p_x}_{x=1}^m$ is a probability distribution. 
We aim to show that $\mN$ is also a thermal operation. 
Since each $\mN_x$ is a thermal operation, there exists an ancillary system $E_x$ such that $\mN_x$ can be written by 
\bal
\mN_x(\rho_A)=\Tr_{E'_x}\qty[U_x(\rho_A\otimes \tau_{E_x})U_x^\dagger],~~~\qty[H_A+H_{E_x},U_x]=0,
\eal
where $\abs{AE_x}=\abs{BE'_x}$, and $U_x$ is an energy-conserving unitary from $AE_x$ to $BE'_x$.

The main idea of the proof in Ref.~\cite[Theorem 17.2.1.]{gour2024resources_of_quantum_world} is to take the direct sum of all the ancillary systems to obtain $\mN$, namely, 
\bal
E=\bigoplus_{x=1}^m E_x,~~E'=\bigoplus_{x=1}^m E'_x.\\
\eal
Furthermore, they take the Hamiltonian of system $E$ so that the thermal state of the system $E$ is represented as
\bal
\tau_E=\bigoplus_{x=1}^mp_x\tau_{E_x},
\label{eq:thermal Gour}
\eal
and define a unitary matrix applied to system $AE$ as
\bal\label{Eq: unitary for convex combination}
U=\bigoplus_{x=1}^m U_x.
\eal
In the proof of Ref.~{\cite[Theorem 17.2.1]{gour2024resources_of_quantum_world}}, they used the fact that the unitary operator defined in Eq.~\eqref{Eq: unitary for convex combination} is energy conserving, i.e., commutes with the Hamiltonian of the $[H_A+H_E, U]=0$ by employing Lemma~17.1.3 of Ref.~\cite{gour2024resources_of_quantum_world}, stating that the unitary operator which commutes with the thermal state of the system also commute with the Hamiltonian to show that the thermal operation
\bal
\rho\mapsto\Tr_{E}\qty[U(\rho\otimes \tau_E)U^\dagger]
\eal
agrees with $\mN=\sum_xp_x\mN_x$. 

In their proof, they use the fact that every state is the Gibbs state for some Hamiltonian, where an explicit form of $H_E$ is not presented. 
Although the existence of such a Hamiltonian is sufficient to prove the convexity, it will still be beneficial to characterize what the desired Hamiltonian looks like.
Here, we explicitly construct the Hamiltonian which has the thermal state represented as Eq.~\eqref{eq:thermal Gour} and commutes with the unitary in Eq.~\eqref{Eq: unitary for convex combination}.

First, note that the state \eqref{eq:thermal Gour} is not the Gibbs state for the most natural Hamiltonian for the system $E=\bigoplus_{x=1}^m E_x$, which is 
\bal
H=\bigoplus_{x=1}^m H_x.
\label{eq:Hamiltonian naive}
\eal

To see this, let us consider an example where two ancillary systems $E_1,~E_2$, are associated with 2- and 3-dimensional Hilbert spaces, respectively, and both have the fully degenerate Hamiltonian with the same energy eigenvalues. 
In this case, the thermal state of each system is
\bal
\tau_{E_1}=\mqty(\frac{1}{2}&0\\0&\frac{1}{2}),~~\tau_{E_2}=\mqty(\frac{1}{3}&0&0\\0&\frac{1}{3}&0\\0&0&\frac{1}{3}).
\eal
If we take $\{p_x\}_x$ as $p_1=p_2=1/2$, the state \eqref{eq:thermal Gour} becomes $\tau_{E}=(\tau_{E_1}\oplus\tau_{E_2})/2=\diag(1/4,1/4,1/6,1/6,1/6)$. However, the Hamiltonian \eqref{eq:Hamiltonian naive} is also fully degenerate, and the thermal state should be the maximally mixed state.

In general, the thermal state for $E$ with Hamiltonian \eqref{eq:Hamiltonian naive} is written as 
\bal
\tau_E=\frac{e^{-\beta H}}{\sum_{x=1}^m Z_{E_x}}=\frac{\sum_{x=1}^m Z_{E_x}\tau_{E_x}}{\sum_{x=1}^m Z_{E_x}},
\eal
where $Z_{E_x}$ is the partition function of $H_x$.
Therefore, to prepare the state $\tau_E$ as a Gibbs state of the system $E$, one needs to modify each Hamiltonian $H_x$ in \eqref{eq:Hamiltonian naive}.

To this end, we define a modified Hamiltonian $\tilde H_x$ of the system $E_x$ by $\tilde H_x = H_x+\Delta_x\mbI$.
The partition function $\tilde Z_{E_x}$ for this Hamiltonian then becomes $\tilde Z_{E_x} = e^{-\beta \Delta_x}Z_{E_x}$, while thermal state $\tau_{E_x}$ remains the same.
Therefore, the thermal state $\tilde\tau_E$ for the modified Hamiltonian $\tilde H \coloneqq \bigoplus_{x=1}^m \tilde H_x$ is obtained as 
\bal
\tilde\tau_E=\frac{\sum_{x=1}^m e^{-\beta \Delta_x}Z_{E_x}\tau_{E_x}}{\sum_{x=1}^m e^{-\beta \Delta_x}Z_{E_x}}.
\eal
By taking $\Delta_x$ as
\bal
\Delta_x =-\frac{1}{\beta}\log\frac{p_x}{Z_{E_x}},
\eal
the state $\tilde\tau_E$ coincides with the desired thermal state in Eq.~(\ref{eq:thermal Gour}). 
By following the same discussion as~\cite[Theorem 17.2.1.]{gour2024resources_of_quantum_world}, we can conclude that the set of thermal operations is convex.

The fact that the pinching channel is thermal operations follows from that the set of the thermal operations is convex by observing that~\cite{Gour_2022_Role_of_quantum_coherence}
\bal
\mP(\rho)=\frac{1}{m}\sum_{x=1}^mU_x\rho U_x^\dagger,~~U_x=\sum_{x'=1}^m e^{\frac{2\pi i x x'}{m}}\Pi_{x'},
\eal
where $m$ is the number of distinct energy levels, and $\Pi_{x'}$ is the projector onto the eigenspace of the Hamiltonian corresponding to the eigenvalue $E_{x'}$.
It can be checked that $U_x$ commutes with the Hamiltonian of the system. Therefore, pinching channel can be represented as a convex combination of thermal operations, which implies that the pinching channel is a thermal operation.

\section{Asymptotic black box work extraction under thermal operations}
\subsection{Construction of work extraction protocol under incoherently conditioned thermal operations}\label{app: incoherently conditioned thermal}
In the following discussion, we consider the i.i.d. black boxes, which contain a finite number of states, i.e., the black box $\qty{\mS^{\rm i.i.d.}_n(S)}_{n=1}^\infty$ with $\abs{S}<\infty$. Our goal in this section is to show that under thermal operations, one can extract the same amount of work asymptotically as the Gibbs-preserving operations and Gibbs-preserving covariant operations when the given state is picked from the i.i.d. black boxes, which contain a finite number of states. 

First, we introduce a new class of operations called incoherently conditioned thermal operations, thermal operations conditioned by the outcome of the incoherent measurement. In Ref.~\cite{Narasimhachar2017resource}, the class of operations called conditioned thermal operations is introduced, in which one performs the measurement on one of the bipartite systems, and applies the thermal operation conditioned by the measurement outcome. The class called incoherently conditioned thermal operation restricts the measurement one can perform to the incoherent measurement. The rigorous definition of the class is the following.
\begin{defn}
    Let $\mH_A,\mH_B,\mH_C$ be Hilbert spaces, and $\mE:\mD(\mH_A\otimes \mH_B)\to \mD(\mH_C)$ be a CPTP map. $\mE$ is called an incoherently conditioned thermal operation when $\mE$ can be decomposed as follows.
    \bal
    \mE=\sum_a\mE^{\rm TO}_a\circ \Lambda^{\rm meas}_a.
    \eal
    Here, $\mE_a^{\rm TO}:\mD(\mH_B)\to\mD(\mH_C),~a=1,2,\ldots,m$ are thermal operations and 
    \bal
    \Lambda^{\rm meas}_a(\rho_{AB}):=\Tr_A\qty[(M^{\rm incoh}_a\otimes I_{B})\rho_{AB}],~a=1,\ldots ,m
    \eal
    be the instruments which represent the incoherent measurement, where $M^{\rm incoh}_a$ is the POVM elements on $\mD(\mH_A)$ which satisfies $\sum_a M^{\rm incoh}_a =I.$ 
    We denote the set of incoherently conditioned thermal operations as ${\rm{ICTO}}(A;B\to C)$ 
    Furthermore, when the measurements are restricted to the incoherent projective measurements, we say that the operation is thermal operation $+$ incoherent projective measurement. We denote the set of thermal operations $+$ incoherent projective measurements as ${\rm{ICPTO}}(A;B\to C)$
\end{defn}

Our first goal is to show the following proposition.
\begin{pro}\label{Pro:EW under TO+cov. meas.}
Let $S\subset \mD(\mH)$ be a subset of density matrices that contain a finite number of density matrices.
    The asymptotic extractable work of the sequence of the i.i.d. black boxes $\qty{\mS_n^{\rm{i.i.d.}}(S)}_{n=1}^\infty$ under incoherently conditioned thermal operations is represented as  
    \bal
    \beta W_{\rm ICTO}(\qty{\mS_n^{\rm{i.i.d.}}(S)}_{n=1}^\infty)=\min_{\rho\in S}D(\rho||\tau).
    \eal
\end{pro}
We remark that all the measurements performed in the work extraction protocol constructed in the following proof are projective. Therefore, we can conclude that the same extractable work can be achieved in the asymptotic regime even with thermal operation $+$ incoherent projective measurement.

\begin{proof}(of ($\leq$) part.)
    To show the ($\leq$) inequality, we note the hierarchy of the operations
    \bal
    {\rm Thermal}\subset {\rm{incoherently ~conditioned~ Thermal}} \subset {\rm{Gibbs-preserving~covariant}} \subset {\rm Gibbs-preserving},
    \eal
    which implies 
\bal
\beta W_{\rm ICTO}(\qty{\mS_n^{\rm{i.i.d.}}(S)}_{n=1}^\infty)\leq\beta W_{\rm GPO}(\qty{\mS_n^{\rm{i.i.d.}}(S)}_{n=1}^\infty)=\min_{\rho\in S}D(\rho||\tau).
\eal
The last inequality is due to Eq.~\eqref{Eq: EW of i.i.d. box under GPO}.
\end{proof}

To show the other inequality, we construct the concrete protocol as follows.
\begin{enumerate}
    \item Given $n$ copies of some unknown state $\rho$, pick up $k_{\delta',p_e}$ copies of states and perform the incoherent measurement. Here, $k_{\delta',p_e}$ is a natural number that depends on the necessary accuracy to identify the initial state.
    \item Identify $\mP(\rho^{\otimes d})$ from the measurement outcome.
    \item Perform the protocol in Ref.~\cite{Brandao_2013_RT_of_quantum_states_out_of} using the information obtained in Step 2.
\end{enumerate}
 One might find it odd that hat the goal of the second step is not to identify $\rho$.
 Actually, it is not possible to perform the state tomography with the incoherent measurement even if the experimenters have an infinite number of copies. 
 The simplest situation is where the experimenters are given a qubit system which is either $\ketbra{+}{+}$ or $\ketbra{-}{-}$ where $\ket{+}=(\ket{0}+\ket{1})/\sqrt{2}$ and $\ket{-}=(\ket{0}-\ket{1})/\sqrt{2}$, and are told to guess which is the state with incoherent measurements. 
We assume that the Hamiltonian of the system is $H=\ketbra{1}{1}$. 
Here, from the definition of the incoherent measurement, for any $\rho\in\mD(\mH)$ and the POVM elements of the incoherent measurement $M^{\rm incoh}_a$ on the system,
\bal
\Tr\qty[M^{\rm incoh}_a \rho]=\Tr\qty[\mP(M^{\rm incoh}_a) \rho] =\Tr\qty[M^{\rm incoh}_a \mP(\rho)]
\eal
holds. This implies that the probability distribution of the outcome of the incoherent measurement does not change when the state is pinched.
Here, we observe the two state $\mP(\ketbra{+}{+}^{\otimes n})$ and $\mP(\ketbra{-}{-}^{\otimes n})$ coincide for every $n\in\mbN$. 
Let $\bs{s},\bs{t}\in\qty{0,1}^n$ be the $n$-bit string. One can see that due to the definition of the pinching channel
\bal
\bra{\bs{s}}\mP(\rho^{\otimes n})\ket{\bs t}\neq 0\Rightarrow \mbox{${\bs{s}} $ and ${\bs{t}}$ belong to the same type class}
\eal
holds.
Furthermore, the direct calculation shows that 
\bal
\bra{\bs{s}}\rho^{\otimes n}\ket{\bs t}=\prod_{i=1}^n\bra{s_i}\rho\ket{t_i},
\eal
where $s_i$ and $t_i$ are the $i$-th alphabets of ${\bs s}$ and ${\bs t}$ respectively. When ${\bs{s}} $ and ${\bs{t}}$ belong to the same type class, 
\bal
\abs{\qty{i\in \{1,\ldots ,d\}| s_i=0,~t_i=1}}=\abs{\qty{i\in \{1,\ldots ,d\}| s_i=1,~t_i=0}}
\eal
holds. 

From this, we can see that all the elements of $\mP(\ketbra{-}{-}^{\otimes n})$ in the energy subspaces are all $1/2^n$, and we can conclude that $\mP(\ketbra{+}{+}^{\otimes n})=\mP(\ketbra{-}{-}^{\otimes n})$ for any $n$. Therefore, even if the experimenters perform any incoherent measurement, they can never distinguish $\ketbra{+}{+}$ and $\ketbra{-}{-}$. 
The discussion above implies that the incoherent measurement is not powerful enough to obtain the full information about the unknown input state $\rho$.

Then, what information can one obtain by the incoherent measurement? In the subsequent discussion, we give an answer to this question.

First of all, we restrict the Hamiltonian in consideration to what satisfies the following property.
\begin{defn}
    Let $H$ be a Hamiltonian, and $E_1,\ldots, E_d$ be the eigenvalues of $H$. We say that $H$ is rationally independent when $H$ satisfies the following.
    \bal
    \sum_iN_iE_i=0, ~\sum_iN_i =0, ~N_i\in\mbZ, ~\forall i\Rightarrow N_i=0,~\forall i.
    \eal
\end{defn}
In the following discussion, we assume that the Hamiltonian of each system satisfies this property. Rational independence prohibits any number of copies of the systems from having degenerate energy levels other than the degeneracy that comes from the permutation of the systems. A simple example is the qutrit system $\mH_3=\Span\qty{\ket{0},\ket{1},\ket{2}}$ and the Hamiltonian of the system $H=E\ketbra{1}{1}+2E\ketbra{2}{2}$. The Hamiltonian $H$ itself does not have degeneracy. However, when we prepare two copies of this, the energy subspace which corresponds to the energy eigenvalue $2E$ of the Hamiltonian of the whole system $H^{\times 2}$ is $\Span\qty{\ket{02}\ket{20},\ket{11}}$. This additional degeneracy comes up because of the rational dependence of the Hamiltonian. 
We note that the energy subspaces of the $n$ copies of the system are spanned by the vectors that belong to the same type class, and there exists a one-to-one correspondence between the energy eigenvalue of the $n$ copies of the system and the type of the eigenvectors. After that, we also extend the discussion to the case where the Hamiltonian is not rationally independent.

Let $\rho\in \mD(\mH)$ be the density operator of the input state, and $\qty{\ket{i}}_{i=1}^d$ be the eigenvectors of the Hamiltonian $H$. Furthermore, we denote $\rho_{ij}:=\bra{i}\rho\ket{j}$.
We show that we can estimate the values called cyclic product defined below in any accuracy by the incoherent measurement.

\begin{defn}
    Let ${\bs s}\in\qty{1,\ldots,d}^m$ be a string of length $m$, which is composed of $m$ different alphabets, i.e., $s_i=s_j\Leftrightarrow i=j$. Here, $m\leq d$ holds. The cyclic product of $\rho$ with respect to the string ${\bs s}$ is defined as
    \bal
    \prod_{i=1}^m\rho_{s_is_{i+1}}~~\qty(=\bra{s_1s_2\cdots s_m}\rho^{\otimes m}\ket{s_2\cdots s_ms_1}),
    \eal
    where we set $s_{m+1}=s_1$.
\end{defn}
Note that the number of the different cyclic products is finite. Furthermore, any cyclic products are in the energy blocks since ${\bs s}=s_1s_2\cdots s_m$ and ${\bs s}'=s_2\cdots s_ms_1$ belong to the same type class. 

\begin{lem}
    Suppose that every matrix element of $\mP(\rho^{\otimes d})$ is given. Then, all the cyclic products of $\rho$ can be calculated from the elements of $\mP(\rho^{\otimes d})$.
\end{lem}
\begin{proof}
    We first consider the diagonal elements of $\rho$. Since for any $i\in\qty{1,\ldots, d}$, $\rho_{ii}\geq 0, \bra{ii\ldots i}\mP(\rho^{\otimes d})\ket{ii\ldots i}\geq 0$ due to the positive semidefiniteness of $\rho$ and $\mP(\rho^{\otimes d})$, and $\bra{ii\ldots i}\mP(\rho^{\otimes d})\ket{ii\ldots i}=\qty(\rho_{ii})^d$ holds, we can calculate $\rho_{ii}$ as $\rho_{ii}=(\bra{ii\ldots i}\mP(\rho^{\otimes d})\ket{ii\ldots i})^{1/d}$.
    One can calculate any cyclic products with respect to the string $\bs s$ with length $m$ by choosing $j\in\qty{1,\ldots,d}$ such that $\rho_{jj}\neq 0$ and noting that
    \bal
    \bra{j\ldots j{\bs s}}\mP(\rho^{\otimes d})\ket{j\ldots j{\bs s'}}&=(\rho_{jj})^{d-\abs{\bs s}}\prod_{i=1}^m\rho_{s_is_{i+1}},\\
    \prod_{i=1}^m\rho_{s_is_{i+1}}&=\frac{\bra{j\ldots j{\bs s}}\mP(\rho^{\otimes d})\ket{j\ldots j{\bs s'}}}{(\bra{jj\ldots j}\mP(\rho^{\otimes d})\ket{jj\ldots j})^{\frac{d-\abs{\bs s}}{d}}}.
    \eal
\end{proof}

Once we obtain the list of all cyclic products, assuming that the Hamiltonian is rationally independent, we can reconstruct $\mP(\rho^{\otimes n})$ for any $n\in\mbN$.
\begin{lem}\label{lemma: any elements can be decomposed to the cyclic products}
    For any positive integer $n$, any nonzero elements of $\mP(\rho^{\otimes n})$ can be represented as the product of cyclic products.
\end{lem}
\begin{proof}
    Due to the definition of the pinching channel, all the matrix elements of $\mP(\rho^{\otimes n})$ that are not inside the energy block are $0$. 
    Therefore, it suffices to consider the matrix elements inside the energy blocks.
    Any nonzero elements of $\mP(\rho^{\otimes n})$ can be written in the form of
    \bal
    \bra{\bs s}\rho^{\otimes n}\ket{\bs t}=\prod_{i=1}^n\rho_{s_it_i},
    \eal
    where ${\bs s},{\bs t}\in\qty{1,\ldots,d}^n$ are the strings of length $n$ which belong to the same type class. This follows from the rational independence of the Hamiltonian.
    One can rearrange $\rho_{s_1t_1},\ldots,\rho_{s_nt_n}$ to the following form.
    \bal\label{eq: the form of matrix elements of the pinched state}
    \rho_{a_1a_2}\rho_{a_2a_3}\cdots \rho_{a_na_1}, ~~~a_1,\ldots,a_n\in\qty{1,\ldots,d}
    \eal
    Note that this does not mean that any matrix elements of $\mP(\rho^{\otimes n})$ are the cyclic products since the $a_1,\ldots,a_n$ can include the duplicated alphabet.
    The existence of such a sequence as Eq.~\eqref{eq: the form of matrix elements of the pinched state} is guaranteed by the assumption that ${\bs s},{\bs t}\in\qty{1,\ldots,d}^n$ are in the same type class.
    Now, we separate this sequence into the cyclic products in the following way. If $a_i=a_j(=\alpha),~ i\neq j$ are the duplicated pair of the alphabets, we divide the sequence above as
    \bal
    \blueunderline{\rho_{a_1a_2}\cdots \rho_{a_{i-1}{\bt \alpha}}} \redunderline{\rho_{{\rt \alpha} a_{i+1}}\cdots \rho_{a_{j-1}{\rt \alpha}}} \blueunderline{\rho_{{\bt \alpha} a_{j+1}} \cdots \rho_{a_n a_1}} \rightarrow  \blueunderline{\rho_{a_1a_2}\cdots \rho_{a_{i-1}{\bt \alpha}}} \blueunderline{\rho_{{\bt \alpha} a_{j+1}} \cdots \rho_{a_n a_1}},~~\redunderline{\rho_{{\rt \alpha} a_{i+1}}\cdots \rho_{a_{j-1}{\rt \alpha}}}
    \eal
These two terms also have the form in Eq.~\eqref{eq: the form of matrix elements of the pinched state}. This procedure can be carried out again and again until the divided terms have no overlaps in the alphabets, in other words, they are divided into the cyclic products. 
This decomposition can be done in any nonzero matrix elements in $\mP(\rho^{\otimes n})$, which completes the proof.
\end{proof}
From these lemmas above, we can easily see the following.
\begin{lem}\label{lemma: d system is sufficient}
For any density matrix of the qudit system $\rho_1,\rho_2\in\mD(\mH)$,
\bal
\mP(\rho_1^{\otimes d})=\mP(\rho_2^{\otimes d})\Leftrightarrow \mP(\rho_1^{\otimes n})=\mP(\rho_2^{\otimes n}),~~\forall n\in\mbN.
\eal
\end{lem}
\begin{proof}
    $(\Leftarrow)$ is obvious, and $(\Rightarrow)$ follows because the left condition implies that all cyclic products of $\rho_1$ and $\rho_2$ are the same, which means the condition of the right-hand side due to Lemma~\ref{lemma: any elements can be decomposed to the cyclic products}.
\end{proof}

Using these lemmata, we can show the Proposition~\ref{Pro:EW under TO+cov. meas.}.

\begin{proof}(of ($\geq$) part.)
In Ref.~\cite{Brandao_2013_RT_of_quantum_states_out_of}, it is shown that one can perform the protocol $\Lambda$ in which for any $\ve'>0$ and $\eta'>0$ there exists $N\in\mbN$ such that
\bal
n\geq N\Rightarrow \exists m_n\in\mbN~~\mbox{s.t.}~~ F(\Lambda\circ\mP(\rho^{\otimes n}),(\ketbra{1}{1}_X,\mu_{m_n}))\geq 1-\ve',~ \abs{D(\rho||\tau)-\frac{1}{n}\log m_n}<\eta'.
\eal
Here, we explicitly denote the thermal state of the battery system together with the excited state of the battery system.
Note that since in this protocol, one first applies the pinching channel, it suffices to obtain information about the pinched state by the incoherent measurement for the work extraction protocol. 
Furthermore, due to Lemma~\ref{lemma: d system is sufficient}, in order to specify the input state, one just needs to perform the quantum state tomography on the pinched $d$ copies of input state $\mP(\rho^{\otimes d})$.

Note that the black box contains a finite number of states. We define $\delta>0$ as
\bal
2\delta :=\min_{\substack{\rho_i,\rho_j\in S\\\mP(\rho_i^{\otimes d})\neq\mP(\rho_j^{\otimes d})}}\| \mP(\rho_i^{\otimes d})-\mP(\rho_j^{\otimes d}) \|_1.
\eal
To perform the quantum state tomography on $\mP(\rho^{\otimes d})$ with accuracy $\delta'$ with respect to the trace distance and with success probability $1-p_e$, it suffices to use 
\bal
k_{\delta',p_e}=\mO\qty(\frac{d^{2d}}{\delta'^2}\log\qty(\frac{1}{p_e}))
\eal

copies of $\mP(\rho^{\otimes d})$~{\cite[Corollary 1.4]{O'Donnell2016efficient}}.
If we take $\delta'$ smaller than $\delta$, the probability of judging $\rho$ as other elements in black boxes $p_e$ can be made arbitrarily small.

Now, we perform the incoherent measurement and the thermal operations conditioned by the measurement outcome as follows.

We denote the appropriate operations for the input state $\rho_i^{\otimes n}$ as $\Lambda_{i,n}$.
We perform state tomography with incoherent measurement using $k_{\delta',p_e}$ copies of given unknown state $\rho$, and obtain an estimate of $\mP(\rho^{\otimes d})$, which we call $\hat E$. We then choose $\tilde i$ such that $\mP(\rho_{\tilde i}^{\otimes d})$ realizes the closest value to $\hat E$, i.e., $\tilde i=\argmin_{i}\|\hat E - \mP(\rho_i^{\otimes d})\|_1$. We then apply $\Lambda_{{\tilde i},n}$ to the rest of the states $\rho^{\otimes n}$. 
Noting that the probability of successfully identifying the unknown state as $\rho_{\tilde i}=\rho_i$ is at least $1-p_e$, this protocol ensures that when the given state is $\rho_i$, the operation applied to $\rho_i^{\otimes n}$ has the form $(1-p_e)\Lambda_{i,n}+p_e\Xi$, where $\Xi$ is some quantum channel. 
This guarantees that for any $i$ and arbitrary $\varepsilon',\eta'>0$, there exists a sufficiently large $n$ such that  
\begin{equation}\begin{aligned}
&F\qty(\qty((1-p_e)\Lambda_{i,n}+p_e\Xi)\circ\mP\qty(\rho^{\otimes n}),(\ketbra{1}{1},\mu_{m_n}) )\geq (1-p_e)(1-\ve')\\
&\abs{D(\rho||\tau)-\frac{1}{n}\log m_n}<\eta'
\end{aligned}\end{equation}

Let us fix arbitrary $\varepsilon>0$ and $\eta>0$.
We choose $k_{\delta',p_e}$ and $n$ to satisfy
\bal
(1-p_e)(1-\ve')\geq 1-\ve.
\eal
Here, the point is that $p_e$ does not depend on $n$ but only on the accuracy of the state tomography $\delta'.$ Furthermore, with respect to the extractable work, 
\bal
\abs{D(\rho||\tau)-\frac{1}{n+k_{\delta', p_e}}\log m_n}&=\abs{ D(\rho||\tau)-\frac{1}{n}\log m_n +\qty(\frac{1}{n}-\frac{1}{n+k_{\delta', p_e}})\log m_n } \\
&\leq \abs{D(\rho||\tau) -\frac{1}{n}\log m_n }+\frac{k_{\delta', p_e}}{n(n+k_{\delta', p_e})}\log m_n\\
&<\eta' +\frac{k_{\delta', p_e}}{n^2}(D(\rho||\tau)+\eta').
\eal
This implies that if we take sufficiently large $n$, we can achieve
\bal
\abs{D(\rho||\tau)-\frac{1}{n+k_{\delta', p_e}}\log m_n}<\eta.
\eal

From these discussions, we can conclude this protocol can achieve the same work extraction as the protocol in Ref.~\cite{Brandao_2013_RT_of_quantum_states_out_of}. 
Since we defined the black box extractable work as the worst-case work extraction, it holds that
\bal
\beta W_{\rm ICTO}(\qty{\mS_n^{\rm{i.i.d.}}(S)}_{n=1}^\infty)\geq\min_{\rho\in S}D(\rho||\tau).
\eal

Even when the Hamiltonian is rationally dependent, we can achieve the same performance. To see this, we consider the pinching channel with respect to the rationally independent Hamiltonian, not the Hamiltonian of the system itself. We denote this channel as $\tilde{\mP}.$
In a similar way as the proof in Appendix~\ref{app: convexity of thermal operation}, one can see that $\Tilde{\mP}$ is a thermal operation.
Due to the definition, $\Tilde{\mP}$ erases the additional degenerate spaces in the pinched density matrix with respect to the original Hamiltonian and satisfies $\Tilde{\mP}\circ \mP=\Tilde{\mP}=\mP\circ\Tilde{\mP}$ (see FIG.~\ref{fig: erasing the additional energy subspace.}).
Note that applying this pinching channel in advance does not increase the extractable work, i.e.,
\bal
\beta W_{\rm ICTO}(\qty{\mS_n^{\rm{i.i.d.}}(S)}_{n=1}^\infty)\geq \beta W_{\rm ICTO}(\qty{\Tilde{\mP}(\mS_n^{\rm{i.i.d.}}(S))}_{n=1}^\infty).
\eal
After we apply this pinching channel, we can apply the same state tomography protocol and thermal operation which follows $\Tilde{\mP}$ conditioned by the result of the tomography. 
Applying the same protocol, we can achieve the same extractable work as the case where the Hamiltonian is rationally independent.
Therefore, it holds that
\bal
\beta W_{\rm ICTO}(\qty{\mS_n^{\rm{i.i.d.}}(S)}_{n=1}^\infty)\geq \beta W_{\rm ICTO}(\qty{\Tilde{\mP}(\mS_n^{\rm{i.i.d.}}(S))}_{n=1}^\infty)\geq \min_{\rho\in S}D(\rho||\tau),
\eal
which completes the proof.
\begin{figure}
    \centering
    \includegraphics[width=9cm]{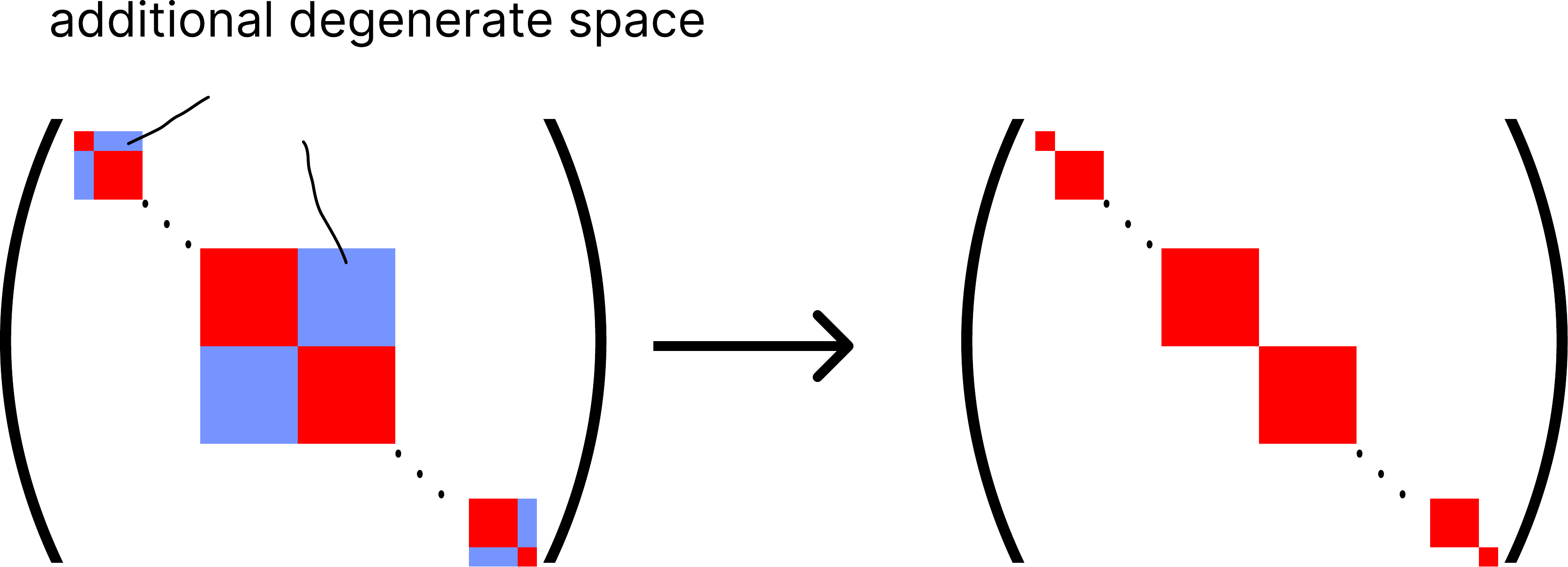}
    \caption{The procedure to erase the additional energy subspace due to the rational dependence of the Hamiltonian of each system.}
    \label{fig: erasing the additional energy subspace.}
\end{figure}
\end{proof}
Note that the measurement is used to perform the quantum state tomography, we use only the projective measurements for the protocol above.

\subsection{Equivalence of thermal operation and incoherently conditioned thermal operation with projective measurements }\label{app:covarianty conditioned equivalent to TO}
In this subsection, we show that one can perform the work extraction protocol mentioned above in thermal operations. We denote the set of thermal operations from $\mD(\mH_A)$ to $\mD(\mH_B)$ as ${\rm TO}(A\to B)$. We start with the following proposition.
\begin{pro}\label{pro: TO=TO+meas.}
Let $A, B$ be the input systems and $C$ be the output system. If $\dim\mH_B=\dim\mH_C$,
    \bal
\rm{TO}(AB\to C)\supset\rm{ICPTO}(A;B\to C)
    \eal
    holds.
\end{pro}
\begin{proof}
    The idea stems from Ref.~\cite[Appendix C]{Lostagilo_Quantum_coherence_time_translation_symmetry}. 
    Since $\mE_i^{\rm{TO}}$ is a thermal operation for any $i\in\qty{1,\ldots,m}$,  $\mE_i^{\rm{TO}}$ can be decomposed as follows.
    \bal
    \mE_i^{\rm{TO}}(\rho_B)=\Tr_{E'_i}\qty[U_i\qty(\rho_B\otimes\tau_{E_i})U_i^\dagger]
    \eal
    Here, $\tau_{E_i}=\exp(-\beta H_{E_i})/Z_{E_i}$ is a thermal state of the ancillary system $E_i$, which is associated with the Hilbert space $\mH_{E_i}$ and the Hamiltonian $H_{E_i}$.
    The unitary operator $U_i$ conserves the energy of the total system, i.e., $\qty[U_i, H_B+H_{E_i}]=0$.
    The system $E'_{i}$ satisfies $B+E_i=C+E'_i$.
    We fix an arbitrary $\mE\in{\rm{ICPTO}}$, and $\mE$ is written as follows.
    \bal\label{eq: the form of the TO+cov.meas. map}
    \mE(\rho_{AB})=\sum_{i=1}^m \Tr_{E'_i}\qty[U_i\qty(\Tr_A\qty[(P^A_i\otimes I_B)\rho_{AB}]\otimes\tau_{E_i})U_i^\dagger]
    \eal
   Consider another map $\Tilde{\mE}:\mD(\mH_A\otimes\mH_B)\to\mD(\mH_C)$, which has the following form.
\bal\label{eq: equivalent map to TO+cov}
\Tilde{\mE}(\rho_{AB})&=\Tr_{A,E'_1,\ldots,E'_m}\qty[\Tilde{U}\qty(\rho_{AB}\otimes \qty(\bigotimes_{i=1}^m \tau_{E_i}))\Tilde{U}^\dagger],\\
\Tilde{U}&=\sum_{i=1}^m P^A_i\otimes U_i. \in\mL\qty(\mH_A\otimes \mH_B\otimes \qty(\bigotimes_i\mH_{E_i})).
\eal
Here, we used the condition $\dim\mH_B=\dim\mH_C$ to make $\tilde{\mE}$ well-defined. Due to this condition, it holds that $\dim\mH_{E_i}=\dim\mH_{E'_i}$ for any $i$, and it follows that
\bal
 \mH_B\otimes\qty(\bigotimes_i\mH_{E_i})=\mH_C\otimes\qty(\bigotimes_i\mH_{E'_i}).
\eal

We can show that $\Tilde{U}$ is indeed a unitary operator and commutes with the Hamiltonian of the whole system $H_{\rm{all}}=H_A+H_B+\sum_iH_{E_i}$ as follows.
\bal
\Tilde{U}\Tilde{U}^\dagger&=\sum_i\sum_j \qty(P^A_i\otimes U_i)\qty(P^A_j\otimes U^\dagger_j)\\
&=\sum_i\sum_j \delta_{ij}P^A_i \otimes U_i U_j^\dagger\\
&=\sum_iP^A_i\otimes I=I.\\
\qty[\Tilde{U},H_{\rm{all}}]&=\sum_i[P^A\otimes U_i,H_A+ H_B+H_{E_i}]\\
&=\sum_i[P^A\otimes U_i,H_A]+\sum_i[P^A\otimes U_i, H_B+H_{E_i}]=0.
\eal
Therefore, the map defined in Eq.~(\ref{eq: equivalent map to TO+cov}) is a thermal operation. Eq.~(\ref{eq: equivalent map to TO+cov}) can be calculated as 
\bal
\Tilde{\mE}(\rho_{AB})&=\Tr_{A,E'_1,\ldots,E'_m}\qty[\Tilde{U}\qty(\rho_{AB}\otimes \qty(\bigotimes_{i=1}^m \tau_{E_i}))\Tilde{U}^\dagger]\\
&=\Tr_{A,E'_1,\ldots,E'_m}\qty[\qty(\sum_{i=1}^m P^A_i\otimes U_i\qty(\rho_{AB})\otimes \qty(\bigotimes_{i=1}^m \tau_{E_i}))\qty(\sum_{j=1}^m P^A_j\otimes U_j)^\dagger]\\
&=\Tr_{A,E'_1,\ldots,E'_m}\qty[\sum_{i=1}^m P^A_i\otimes U_i\qty(\rho_{AB}\otimes \qty(\bigotimes_{i=1}^m \tau_{E_i})) \qty(P^A_i\otimes U_i)^\dagger]\\
&=\sum_{i=1}^m\Tr_{A,E'_i}\qty[P^A_i\otimes U_i\qty(\rho_{AB}\otimes \tau_{E_i}) \qty(P^A_i\otimes U_i)^\dagger]\\
&=\sum_{i=1}^m\Tr_{E'_i}\qty[\Tr_A \qty{P^A_i\otimes U_i\qty(\rho_{AB}\otimes \tau_{E_i}) \qty(P^A_i\otimes U_i)^\dagger}]\\
&=\sum_{i=1}^m \Tr_{E'_i}\qty[U_i\qty(\Tr_A\qty[(P^A_i\otimes I_B)\rho_{AB}]\otimes\tau_{E_i})U_i^\dagger],
\eal
and we obtain Eq.~(\ref{eq: the form of the TO+cov.meas. map}).
\end{proof}

\begin{figure}
    \centering
    \includegraphics[width=12cm]{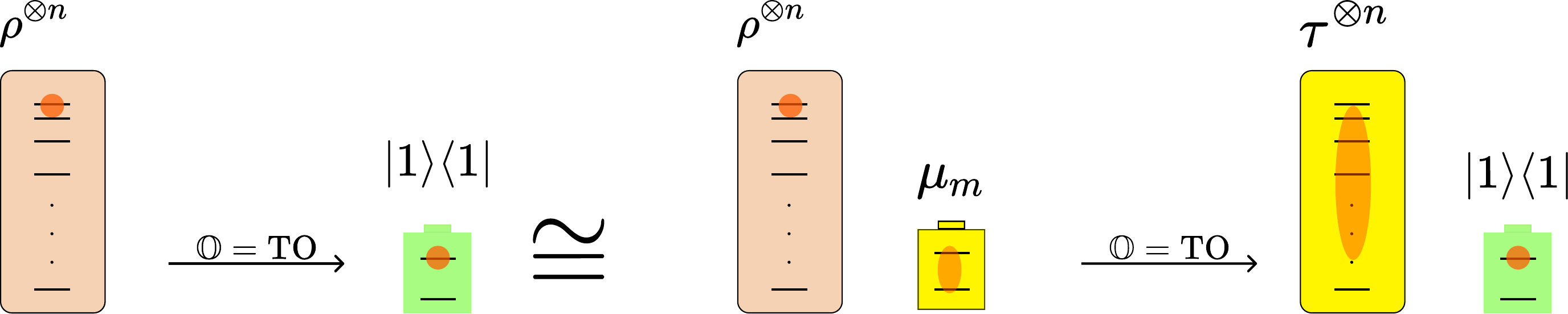}
    \caption{In our work extraction setting, we can append the thermal states to the initial and output states to make the dimensions of the input and output systems the same.}
    \label{TO =TO+meas.}
\end{figure}
Since tracing out the subsystems and adding other thermal states are free operations, we can take the dimensions of the input and output systems equally, and we can undo this by tracing out the added systems, this procedure does not affect the extracted work~(see FIG.~\ref{TO =TO+meas.}).  
Therefore, we can take the input and output system so that the condition $\dim\mH_B=\dim\mH_C$ is satisfied, which implies that we can carry out the protocol mentioned above by a thermal operation.

Combining Proposition~\ref{Pro:EW under TO+cov. meas.} and Proposition~\ref{pro: TO=TO+meas.}, we reached the following theorem.
\begin{thm}[Theorem~\ref{thm:asymptotic TO} in the main text]
The asymptotic extractable work of the sequence of the  i.i.d. black boxes $\qty{\mS_n^{\rm{i.i.d.}}(S)}_{n=1}^\infty$ satisfying $\abs{S}<\infty $ under thermal operations is 
    \bal
    \beta W_{\rm TO}\qty(\qty{\mS_n^{\rm{i.i.d.}}(S)}_{n=1}^\infty)=\min_{\rho\in S}D(\rho||\tau).
    \eal
\end{thm}

\section{Black box channel resource distillation}\label{app: black box channel distillation}

In this section, we introduce a more general task of \emph{black box channel resource distillation}.
We extend black box work extraction to the setting of general resource theories of quantum channels~\cite{Takagi_2019_general_resource,Gour_2019_how_to,Liu_2019_resource_theories,Liu_2020_operational_resource,Regula_2021_one-shot,Regula_2021_fundamental_limitations,Fang_2022_no-go_theorems,Kuroiwa_2024_robustness}, where we do not specify the details of the physical setting to focus on the mathematical structure that different physical settings share. 
Let $\mbF,\mbF'$ be the sets of free states, the states that can be prepared without any costs, and $\mbO,\mbO'$ be the set of free operations, the operations applicable with no cost.
A state that is not a free state is called a resource state. 
The minimum constraint for these two sets is that any free operations map free states to free states, i.e., $\forall \Lambda\in\mbO,~\Lambda(\mbF)\subset\mbF'$. 
The set of physically implementable operations is often smaller than the set of those satisfying this minimum constraint. 

Some of the resource states are essential to implementing useful protocols. For instance, the Bell state is used for tasks such as quantum teleportation~\cite{Bennett_teleporting_unknown} and superdense coding~\cite{Bennett_communication_via}. However, it is difficult to share the Bell state between two parties at a distance in advance since the quantum state is exposed to noise during transportation. 
To obtain a state very close to the Bell state, one needs to apply a limited class of operations, such as local operations and classical communications (LOCC), to distill as many copies of the Bell states as possible from the noisy entangled states shared in advance.
Such a procedure to extract useful resource states from noisy resource states is called \emph{(state) resource distillation}~\cite{Bennett_1996_concentrating,Bravyi_2005_universal,Brandao_2013_RT_of_quantum_states_out_of,Winter_2016_operational,Marvian_2020_coherence}. 
Work extraction is the special case of the resource distillation task in the resource theory of thermodynamics.

The resource distillation task can also be extended to the channel version. Suppose we are accessible to a channel that is not free to operate. 
One can apply some pre-and post-processing to convert the given channel to another channel. 
Such a map that maps a channel to channel is called a superchannel.
This procedure is also observed in the framework of the quantum resource theory in which the set $\mbS$ of free superchannels, which can be applied with no cost, is also limited. The minimal constraint for the free superchannel is that free superchannel maps any free channel in $\mbO$ to another free channel in $\mbO'$, i.e., $\forall \Theta\in\mbS,~\Theta(\mbO)\subset\mbO'$. Again, the constraint from the physical setting determines the sets $\mbO$ and $\mbO'$, and also limits the range of the free superchannel.
When the channel in possession is not the target channel, i.e., the channel one wants to apply, one seeks the free superchannel so that the mapped channel exhibits a close action to that of the target channel. 
Such a task is an extended version of state resource distillation, called a \emph{channel resource distillation}~\cite{Wang_2019_resource,Wang_2019_quantifying,Gour_2020_dynamical_entanglement,Takagi_2020_application,Liu_2019_resource_theories,Liu_2020_operational_resource,Regula_2021_one-shot, Kim_2021_One-shot_manipulation, Takagi_2022_One-shot_yield-cost_relation,Regula_2021_fundamental_limitations,Fang_2022_no-go_theorems}. 
If one takes the given channel and the target channel as a state preparation channel, the channel resource distillation task is reduced to the state resource distillation task.

We formulate the black box channel resource distillation and clarify the relationship between the performance of the black box channel resource distillation task and the composite hypothesis testing divergence.

\subsection{Preliminaries}
We employ a general framework of channel resource manipulation introduced in Ref.~\cite{Regula_2021_one-shot}. In our discussion, we assume that the free superchannels $\mbS$ satisfy the minimal constraint $\forall \Theta\in\mbS,~\Theta(\mbO)\subset\mbO'$, and the sets $\mbO$ and $\mbO'$ of free operations are closed and convex.
We denote the set of all the channels we consider as $\mbO_{\rm all}$.
Let us define some quantities that were originally defined in the states and extended to the channel.

Let $\mH_A$ be the input space of $\mbO_{\rm all}$ and $\mH_R\cong\mH_A$ be the Hilbert space of the ancillary system.
We also define $\mD(\mH)$ to be the set of all states acting on $\mH$. 
The fidelity between channels $\mE_1,~\mE_2\in\mbO_{\rm all}$ is defined as 
\bal
F(\mE_1,\mE_2)=\min_{\rho\in\mD(\mH_R\otimes \mH_A)}F({\rm id}\otimes \mE_1(\rho),{\rm id}\otimes \mE_2(\rho)),
\eal
where $F(\rho,\sigma)=\| \sqrt{\rho}\sqrt{\sigma} \|_1^2$ is the square fidelity. We can also define the resource measure called robustness, originally introduced in the context of an entanglement measure~\cite{Vidal_robustness_of_entanglement}, for a given channel.
Let $\Tilde{\mbO}$ be a set satisfying $\mbO\subset\Tilde{\mbO}\subset\mbO_{\rm all}$. The robustness of the channel $\mE$ is defined as 
\bal
R_{\mbO;\Tilde{\mbO}}=\inf \lset r\sbar \frac{\mE+r\mM}{1+r}\in\mbO,~\mM\in\Tilde{\mbO}\rset.
\eal
When we take $\Tilde{\mbO}=\mbO$, the quantity above is called the standard robustness $R_s(\mE)$, and when we take $\Tilde{\mbO}=\mbO_{\rm all}$, the quantity is called the generalized robustness $R_g(\mE)$.
Also, we define the smoothed robustness as
\bal
R^\ve_{\mbO;\Tilde{\mbO}}(\mE)=\min\lset R_{\mbO;\Tilde{\mbO}}(\mE') \sbar F(\mE,\mE')\geq 1-\ve \rset.
\eal
The log robustness is also defined as follows.
\bal
LR_{\mbO;\Tilde{\mbO}}(\mE)=\log \qty(1+R_{\mbO;\Tilde{\mbO}}(\mE)),~~~LR^\ve_{\mbO;\Tilde{\mbO}}(\mE)=\log \qty(1+R^\ve_{\mbO;\Tilde{\mbO}}(\mE))
\eal

Another important quantity is the hypothesis testing divergence. We consider the composite version of this quantity. Let $\mbB\subset \mbO_{\rm all}$ be the black box. The hypothesis testing divergence of $\mbB$ with respect to $\mbO$ is defined as 
\bal
D^\ve_H(\mbB||\mbO)&:=
\max_{\rho\in\mD(\mH_R\otimes \mH_A)}D^\ve_H(\id\otimes\mbB(\rho)||\id\otimes\mbO(\rho))\\
&=\max_{\rho\in\mD(\mH_R\otimes \mH_A)}\qty(-\log
\min_{\substack{0\leq M\leq I\\\max_{\mE\in\mbB}\Tr[{\rm id}\otimes\mE(\rho)(I-M)]\leq \ve}}
\max_{\mF\in\mbO}
\Tr[{\rm id}\otimes \mF(\rho)M]).
\eal
Since $\Tr[{\rm id}\otimes \mF(\rho)M]$ is linear with respect to $\mF, ~\rho$, and $M$, we can exchange the min and max because of Sion's minimax theorem~\cite{sion_1958_general_minimax_theorems}.
Note that the maximization with respect to the state can be replaced with the maximization with respect to the pure state since we can add an ancillary system to purify the input mixed state, which does not decrease the performance of the hypothesis testing.
Specifically, the hypothesis testing divergence with $\ve=0$ is called a min-divergence~\cite{Datta_2009_min_and_max_relative} and denoted as $D_{\min}$.

If it holds that $\Span(\mbO)=\mbO_{\rm all}$, we say that the set $\mbO$ is \emph{full dimensional}, and otherwise $\mbO$ is \emph{reduced dimensional}.
When $\mbO$ is full dimensional, the standard robustness is guaranteed to be finite for any channels $\mE\in\mbO_{\rm all}$. On the other hand, when $\mbO$ is reduced dimensional, the standard robustness is not finite in general. To characterize the distillation tasks in such cases, we need to consider the affine hull $\aff(\mbO)$ of the set of free operations. In that case, the performance of the distillation task is often measured in the quantities such as $D^\ve_H(\cdot||\aff(\mbO))$, and $D_{\rm min}(\cdot||\aff(\mbO))$.

\subsection{One-shot black box channel resource distillation}
In the following discussion, we formulate the black box channel resource distillation framework. Let $\mbB\subset\mbO_{\rm all}$ be a subset of channels, and $\mbT\subset\mbO_{\rm all}$ be the set of target resource channels, i.e., the resourceful channels one desires to apply. In the black box channel distillation framework, we apply the free superchannel $\Theta \in \mbS$ to the given channel picked from the black box $\mbB$. The experimenters know that the channel is an element of the black box but are not told which channel is given. 
Our goal is to determine what reference channels can be distilled from the unknown channel from the black box and the allowed superchannels.
We assume that the target channel $\mN$ maps the input pure state to the pure state. 
Then, the following holds.
\begin{pro}\label{pro: converse part for channel distillation}
    Let $\mbB\subset \mbO_{\rm all}$ be a subset of the channels, and $\mN\in\mbO_{\rm all}$ be the target channel such that $\id\otimes\mN(\psi)$ is pure for any pure state $\psi\in\mD(\mH_R\otimes\mH_A)$. 
    If there exists a free superchannel $\Theta\in \mbS$ s.t. $\min_{\mE\in\mbB}F(\Theta(\mE),\mN)\geq 1-\ve$, it holds that
    \bal
    D^\ve_H(\mbB||\mbO)\geq D_{\rm min}(\mN||\mbO), ~\mbox{and}~~D^\ve_H(\mbB||\aff(\mbO))\geq D_{\rm min}(\mN||\aff(\mbO))
    \eal
\end{pro}
\begin{proof}
    One can show the data processing inequality of the composite hypothesis testing divergence in the same way as the discussion above; it holds that
    \bal
    D^\ve_H(\mbB||\mbO)&\geq D^\ve_H(\Theta (\mbB)||\Theta(\mbO))\geq D^\ve_H(\Theta(\mbB)||\mbO').
    \eal
Here, we can rewrite the last line as 
\bal
D^\ve_H(\Theta(\mbB)||\mbO')&=-\log \min_{\rho\in\mD(\mH_R\otimes \mH_{A'})}
\min_{\substack{0\leq M\leq I\\\max_{\mE\in\mbB}\Tr[{\rm id}\otimes\Theta(\mE)(\rho)(I-M)]\leq \ve}}
\max_{\mF\in\mbO'}
\Tr[{\rm id}\otimes \mF(\rho)M]\\
&=-\log\max_{\mF\in\mbO'} \min_{\rho\in\mD(\mH_R\otimes \mH_{A'})}
\min_{\substack{0\leq M\leq I\\\max_{\mE\in\mbB}\Tr[{\rm id}\otimes\Theta(\mE)(\rho)(I-M)]\leq \ve}}
\Tr[{\rm id}\otimes \mF(\rho)M]\\
&=-\log\max_{\mF\in\mbO'} \min_{\psi\in {\rm PURE}(\mH_R\otimes \mH_{A'})}
\min_{\substack{0\leq M\leq I\\\max_{\mE\in\mbB}\Tr[{\rm id}\otimes\Theta(\mE)(\psi)(I-M)]\leq \ve}}
\Tr[{\rm id}\otimes \mF(\psi)M].
\eal
In the third line, we purified the fed state. Due to the assumption that ${\rm id}\otimes \mN(\psi)$ is always pure for any pure input state $\psi$, We can check that $M={\rm id}\otimes \mN(\psi)$ is feasible as follows.
\bal
\min_{\mE\in\mbB}\Tr\qty[\qty({\rm id}\otimes \mN(\psi))\qty({\rm id}\otimes \Theta(\mE)(\psi))]&=\min_{\mE\in\mbB}F\qty({\rm id}\otimes \mN(\psi),{\rm id}\otimes \Theta(\mE)(\psi))\\
&\geq \min_{\mE\in\mbB}\min_{\psi\in{\rm PURE}(\mH_R\otimes \mH_{A'})}F\qty({\rm id}\otimes \mN(\psi),{\rm id}\otimes \Theta(\mE)(\psi))\\
&=\min_{\mE\in\mbB}F(\mN,\Theta(\mE))\geq 1-\ve.
\eal
From this, it holds that
\bal
D^\ve_H(\mbB||\mbO)&\geq D^\ve_H(\Theta(\mbB)||\mbO')\\
&\geq -\log\max_{\mF\in \mbO'}\min_{\psi\in{\rm PURE}(\mH_R\otimes \mH_{A'})}\Tr\qty[({\rm id}\otimes \mF(\psi))({\rm id}\otimes \mN(\psi))]\\
&=\min_{\mF\in \mbO'}\max_{\psi\in{\rm PURE}}D_{\rm min}({\rm id}\otimes \mN(\psi)||{\rm id}\otimes \mF(\psi))\\
&=D_{\rm min}(\mN||\mbO').
\eal
The second inequality can be shown similarly.
\end{proof}

The achievability part can be shown as follows.
\begin{pro}\label{pro: direct part for channel distillation}
    Let $\mbB\subset\mbO_{\rm all}$ be the subset of the channels. When $\mbB$ satisfies 
    \bal
    D^\ve_H(\mbB||\mbO)\geq LR^\delta_{s}(\mN)\mbox{ or } D^\ve_H(\mbB||\aff(\mbO))\geq LR^\delta_{g}(\mN)
    \eal
    there exists a superchannel $\Theta:\mbO_{\rm all}\to \mbO'_{\rm all}$ which satisfy
    \bal
    \min_{\mE\in\mbB}F(\Theta(\mE),\mN)\geq 1-\ve-\delta.
    \eal
\end{pro}
\begin{proof}
We consider the situation where $\mbB$ satisfies $D^\ve_H(\mbB||\mbO)\geq LR^\delta_{s}(\mN)$. The statement in the other situation can also be shown similarly.
    Let $\Tilde{\mN}$ be the channel which satisfies $LR_s(\tilde{\mN})=LR^\delta_s(\mN)$, $F(\mN,\Tilde{\mN})\geq 1-\delta$. From the definition, one can see that there exists a free channel $\mQ\in\mbO$ which satisfies
    \bal
    \frac{\Tilde{\mN}+\qty(2^{LR^\delta_s(\mN)}-1)\mQ}{2^{LR^\delta_s(\mN)}}\in\mbO.
    \eal
    Note that, due to the convexity of the set of free operations, it holds that 
    \bal\label{Eq: condition channel is free}
    r\geq 2^{LR^\delta_s(\mN)}-1~\Rightarrow ~& \frac{\Tilde{\mN}+r\mQ}{1+r}=\frac{2^{LR^\delta_s(\mN)}}{1+r}\frac{\Tilde{\mN}+\qty(2^{LR^\delta_s(\mN)}-1)\mQ}{2^{LR^\delta_s(\mN)}}+\frac{r-2^{LR^\delta_s(\mN)}+1}{1+r}\mQ\in\mbO.
    \eal

    From the definition of the hypothesis testing divergence, it holds that
    \bal
    D^\ve_H(\mbB||\mbO)&=\max_{\rho\in\mD(\mH_R\otimes \mH_A)}\qty(-\log
\min_{\substack{0\leq M\leq I\\\max_{\mE\in\mbB}\Tr[{\rm id}\otimes\mE(\rho)(I-M)]\leq \ve}}
\max_{\mF\in\mbO}
\Tr[{\rm id}\otimes \mF(\rho)M])\\
&=\max_{\psi\in\mD(\mH_R\otimes \mH_A)}\qty(-\log
\min_{\substack{0\leq M\leq I\\\max_{\mE\in\mbB}\Tr[{\rm id}\otimes\mE(\psi)(I-M)]\leq \ve}}
\max_{\mF\in\mbO}
\Tr[{\rm id}\otimes \mF(\psi)M]).
    \eal
    In the second line, we purified the input state.
    Let $\psi^*, M^*$ be the pure state and the POVM element which achieves the optimal performance on the hypothesis testing tasks.
    It holds that
    \bal\label{Eq: bound of the free channels}
    \Tr[{\rm id}\otimes \mF(\psi^*)M^*]\leq 2^{-D^\ve_H(\mbB||\mbO)},~~\forall \mM\in\mbO.
    \eal
    
    Here, we construct the superchannel $\Theta$ as follows.
    \bal
    \Theta(\mL)=\Tr\qty[M^* ({\rm id}\otimes \mL(\psi^*))]\Tilde{\mN}+\Tr\qty[(I-M^*) ({\rm id}\otimes \mL(\psi^*))]\mQ.
    \eal
    When the input channel is $\mF\in\mbO$,
    \bal
    \mF\in\mbO~\Rightarrow~\Theta(\mF)=\Tr\qty[M^* ({\rm id}\otimes \mF(\psi^*))]\Tilde{\mN}+\Tr\qty[(I-M^*) ({\rm id}\otimes \mF(\psi^*)]\mQ.
    \eal
    Here, From Eq.~\eqref{Eq: bound of the free channels} and the assumption, it holds that
    \bal
    \frac{\Tr\qty[(I-M^*) ({\rm id}\otimes \mF(\psi^*)]}{\Tr\qty[M^* ({\rm id}\otimes \mF(\psi^*))]}&\geq \frac{1-2^{-D^\ve_H(\mbB||\mbO)}}{2^{-D^\ve_H(\mbB||\mbO)}}\\
    &\geq \frac{1-2^{-LR^\delta_s(\mN)}}{2^{{-LR^\delta_s(\mN)}}}=2^{LR^\delta_s(\mN)}-1.
    \eal
    Because of Eq.~\eqref{Eq: condition channel is free}, one can see that $\Theta$ is a free superchannel.

    Furthermore, for any input channel $\mE\in\mbB$, it holds that
    \bal
    F(\Theta(\mE),\mN)&=\min_{\psi\in{\rm PURE}(\mH_R\otimes \mH_{A'})}F(\Theta(\mE)(\psi),\mN(\psi))\\
    &\geq\Tr\qty[M^* ({\rm id}\otimes \mE(\psi^*))]F(\Tilde{\mN},\mN)\\
    &\geq (1-\ve)(1-\delta)\geq 1-\ve-\delta,
    \eal
    which completes the proof.
\end{proof}

\subsection{One-shot distillable resource}
To extend the black box distillation task in the channel distillation setting, we consider the set of the reference channel, $\mbT\subset \mbO_{\rm all}$.
We define the distillable resource of the black box $\mbB$ as follows.
\begin{defn}
    The one-shot distillable resource of the black box $\mbB$ with respect to the resource measure $\mathfrak{R}$ is defined as 
    \bal
    d^\ve_{\mbS,\mathfrak{R}}(\mbB):=\max\lset \mathfrak{R}_\mbO(\mT) \sbar \mT\in\mbT, ~\exists \Theta\in\mbS~{\rm s.t.}~\min_{\mE\in\mbB}F(\Theta(\mE),\mT)\geq 1-\ve \rset .
    \eal
\end{defn}

Due to Proposition~\ref{pro: converse part for channel distillation} and Proposition~\ref{pro: direct part for channel distillation}, the following holds.
\begin{pro}
    Let $\mbB\subset\mbO_{\rm all}$ be the black box of the channels, and $\mbT\subset\mbO_{\rm all}$ be the reference channels. Furthermore, assume that any channel $\mT\in\mbT$ satisfies 
    \bal
    D_{\rm min}(\mT||\mbO)=LR_s(\mT).
    \eal
    Then, the distillable resource with respect to the min-divergence $D_{\rm min }$ of the black box $\mbB$ satisfies
    \bal
    d^\ve_{\mbS,D_{\rm min}}(\mbB)=\max\qty{D_{\rm min}(\mT||\mbO)~|~\mT\in\mbT, D^\ve_H(\mbB||\mbO)\geq D_{\rm min}(\mT||\mbO)}
    \eal

    Furthermore, when any channel $\mT\in\mbT$ satisfies 
    \bal\label{app Eq: Dmin=generalized robustness}
    D_{\rm min}(\mT||\aff(\mbO))=LR_g(\mT),
    \eal
    the distillable resource with respect to the min-divergence $D_{\rm min}$ of the black box $\mbB$ satisfies
    \bal
        d^\ve_{\mbS,D_{\rm min}}(\mbB)=\max\qty{D_{\rm min}(\mT||\aff(\mbO))~|~\mT\in\mbT, D^\ve_H(\mbB||\aff(\mbO))\geq D_{\rm min}(\mT||\aff(\mbO))}
    \eal
\end{pro}

To recover the same result as Theorem~\ref{app theorem: one shot extractable work under GPO}, it suffices to take $\mbB$ as the set of the preparation channels that prepares the states in the set $\mS\subset\mD(\mH)$ of states, $\mbO$ as the set of thermal state preparation channels, and $\mbT$ as 
\bal
\mbT=\qty{\mT_m:\mD(\mbC)\to\mD(\mH_X), (1)\mapsto (\ketbra{0}{0},\mu_m)~|~m\geq 0}.
\eal
One can see that this set of channels satisfies the condition in~Eq.~\eqref{app Eq: Dmin=generalized robustness}, and it holds that
\bal
D_{\rm min}(\mT||\aff(\mbO))=LR_g(\mT)=\log m.
\eal
From this, it holds that
 \bal
        d^\ve_{\mbS,D_{\rm min}}(\mbB)&=\max\qty{D_{\rm min}(\mT||\aff(\mbO))~|~\mT\in\mbT, D^\ve_H(\mbB||\aff(\mbO))\geq D_{\rm min}(\mT||\aff(\mbO))}\\
        &=D^\ve_H(\mbB||\mbO)\\
        &=D^\ve_H(\mS||\tau),
\eal
which agrees with the result in Theorem~\ref{app theorem: one shot extractable work under GPO}.

\end{document}